\documentclass[USenglish, cleveref, autoref, thm-restate]{llncs}
\usepackage{xurl}
\usepackage{amsmath, amssymb, amsfonts}
\usepackage{thmtools}

\usepackage{setspace}
\usepackage{comment}
\usepackage{float}
\usepackage[hidelinks]{hyperref}
\usepackage{bookmark}
\usepackage[utf8]{inputenc}

\usepackage[T1]{fontenc}
\usepackage{graphicx}
\usepackage{mathtools}
\usepackage{cleveref}
\usepackage{subcaption}
\usepackage{algorithm,algorithmic}
\usepackage{tikz,pgfplots}
\pgfplotsset{compat=1.18}
\usepackage{xcolor}
\usepackage{graphicx}
\declaretheorem[name=Proposition]{prop}
\crefname{prop}{Proposition}{Propositions}
\Crefname{prop}{Proposition}{Propositions}

\begin{document}

\author{
    Yunqi Zhang \and 
    Shaileshh Bojja Venkatakrishnan 
}


\institute{The Ohio State University\\
\email{\{zhang.8678, bojjavenkatakrishnan.2\}@osu.edu}}

\title{
Rethinking Collaborative Trust for Verifiably Decentralized Blockchain Systems}
%
%
%
\maketitle              

\begin{abstract}
Despite the promise of decentralization, measurement studies have identified a conspicuous lack of decentralization in blockchains. 
Centralization has been observed in almost all layers of the blockchain, in decentralized applications, and in decentralized autonomous organizations. 
In many cases, it is practically impossible to definitively determine the extent of centralization in the system. 
While multiple works have proposed methods to decrease centralization, by and large blockchains continue to be significantly centralized. 

In this paper, we develop a general framework for building verifiably decentralized blockchain systems. 
Our framework is motivated by the core observation that the richness and diversity of collaborative interactions between users---rather than resource uniformity---captures the essence and extent of decentralization in a blockchain system. 
Existing blockchains do not have any incentive mechanisms to encourage inter-coalition collaboration, which directly contributes to centralization.
We propose a novel reward design that incentivizes users to collaborate with other users without forming isolated coalitions. 
Technically, our method uses a Sybil-resistant asymmetric Shapley value for reward attribution within a collaboration group, and the theory of expander graphs for measuring and enforcing decentralization. 

Our framework is general and can be adapted to alleviate centralization in any layer, application, or decentralized organization. 
It also has important implications beyond the topic of centralization.
For example, we show that our solution can naturally address the blockchain scalability problem. 
We also identify a new class of decentralized collaborative applications that have hitherto been unexplored in blockchains. 
\end{abstract}
\section{Introduction}

In its more than 15 years of existence, blockchains have matured significantly as a flexible, open, transparent, and secure platform for hosting a wide array of decentralized applications~\cite{nakamoto2008bitcoin}. 
Compared to the first blockchain---Bitcoin---today's blockchains excel at achieving superior transaction throughput, lower confirmation latency, support for diverse applications, more efficient resource consumption and sharing, and improved community governance.  
This has been possible in large part due to the intensive research from both academia and the industry, which has resulted in innovations in new consensus protocols, layer-2 scaling methods, cryptography, incentive design, networking, and beyond~\cite{gervais2016security,bonneau2015sok}. 
Today blockchains have a market capitalization exceeding 2 trillion dollars, and millions of daily users worldwide. 

Despite the stellar progress, blockchains remain a niche technology enthusiastically adopted by a small section of the consumer population and mostly ignored by the rest~\cite{blockchainadoptionstats}.
The reasons for this behavior are complex and multi-faceted;
the poor scalability of blockchains is often cited as one of the main technological reasons for its practical limitations. 
However, in many blockchains adoption remains poor (relative to centralized counterparts) even after significantly improving scalability, cost, and performance efficiency~\cite{cloud2026}. 
A noteworthy example is from the domain of decentralized physical infrastructure networks (DePINs). 
Due to its crowd-sourced design, many DePINs provide resources---including essential resources such as compute, storage, network, etc.---to end-users at a cost that is a fraction (e.g., 1/10-th) of the cost of the same resource from a centralized provider~\cite{wu2025degrees}.
And yet, it is the centralized services that enjoy the larger user base across the board and by far~\cite{depin2026}.
In addition to protocol improvements, the blockchain industry has also innovated on novel applications. 
E.g., due to the growing demand for AI, in recent years a number of blockchains have focused (or, in some cases, pivoted to) on providing decentralized AI services.
Time and again such innovations are touted as the next ``killer app'' for blockchains. 
Whether true or not, the remarks seem to suggest that  perhaps a ``killer app'' for blockchains does not exist as of yet~\cite{privacykillerapp26,stablecoinkillerapp,frameskilerapp}.  

In this paper  we present a novel blockchain incentive paradigm that enables new types of decentralized applications while improving existing applications. 
At the heart of our proposal is a fundamental redefinition of what    ``decentralization'' means in blockchains. 
It is widely accepted that blockchains are transparent, secure, and decentralized. 
Transparency is by design: all blocks are publicly visible and can be verified by anyone. 
Many blockchains also provide strong mathematical guarantees for the security of their systems.
However, the third property---decentralization---is tricky to verify or provide guarantees for.
It is fundamentally challenging (and practically impossible) to verify whether a group of nodes is colluding.
This is true even if there are strong signals indicative of decentralization such as:  the network is open and permissionless, there are thousands of worldwide nodes running on independent infrastructure providers, key governance decisions are made democratically, 
the amount of resource (such as stake) owned by any one node is a tiny fraction of the overall resources available on the chain,
historical measurements show a lack of distinct clusters on the network topology, and there are no  catastrophic chain reversal or other fatal attacks.  
Therefore, decentralization in blockchains today is really more of a commonly held {\em belief} (or, not!) than a verifiable property. 

In blockchains, collusion is often viewed as an undesirable behavior that enables attacks detrimental to chain security~\cite{meaningofdecentralization}. 
Collusions are also at the heart of why blockchains become centralized. 
We focus on a particularly important type of a collusion in which a coalition of nodes---over time---pool their resources, actively collaborate, improve service efficiency, and slowly dominate the marketplace (e.g., the block mining market) making it difficult for smaller players with fewer resources to compete. 
We claim that this commonplace phenomenon---which we term {\em ossification}---is a dominant factor in determining how decentralized a blockchain is.  
Ossification is rampant in today's blockchains. 
A familiar example is the emergence of mining pools, wherein nodes have an incentive to join a pool and lower reward variance rather than operate  alone. 

Ossification is beneficial for an economy outside of blockchains. 
It is the process by which centralized organizations form partnerships, invest in the partnerships, and over time grow to become industry behemoths. 
To ossify is the natural driving force in an economy. 
It leads to streamlining of processes, improving efficiency, leveraging the economy of scale and network effects to provide services (or, goods) at a lower cost.
However, in blockchains ossification is a centralizing force that must be resisted. 
Unfortunately, today's blockchains have no mechanisms to resist ossification and have largely succumbed to the ossifying forces across the protocol stack. 

We define a decentralized system as a system where there is no incentive for any subset of nodes to ossify. 
To make a blockchain system ossification-resistant requires new incentive mechanisms beyond what is available in existing blockchains. 
We propose a novel reward mechanism in which nodes are rewarded for {\em collaboratively} performing tasks (such as mining a block, or providing an application's service) with a diverse set of users.
The reward mechanism is Sybil resistant: a party cannot pretend to  collaborate with others when it is really just collaborating with itself. 
Importantly, our mechanism encourages collaboration with diverse entities over time while penalizing static collusion sets. 
By defining collaborations as a rooted directed acyclic graph (DAG), we compute rewards as an asymmetric Shapley value with carefully chosen DAG-dependent weights to provide Sybil resistance.
Unlike existing blockchains where collaborations (e.g., committee assignments) are algorithmically computed, we leave the choice of collaboration and collaborators as a subjective decision made solely by the users.  
Any faults in the provided service leads to the slashing of only the collaboration initiator and not the individual collaborators. 

The implications of our mechanism design are many. 
It is naturally ossification-resistant as a  group of nodes that consistently collaborate only with each other and not with other nodes will be penalized.
Without ossification, even small players (e.g., nodes with a relatively low stake) can compete by forming collaborations with other nodes. Moreover, it is in the interest of the nodes to be as inclusive as possible in their collaborations. 
Instead of an economy of competition which leads to centralization, we can have an economy of inclusiveness, collaboration, and community. 

Forbidding ossification does not mean innovation and efficiency within the system would stagnate.  
Nodes still compete with each other which encourages innovation, except now all the other nodes are also involved in the competition via collaborations. 
However, we highlight that there are important differences in the type of services suitable for an unossified economy compared to an ossified one.  
In a successful ossified coalition, participants have deep-rooted trust relationships with each other, lower communication barriers, and highly optimized service methods  developed from years of experience in the domain.
Conversely, services where efficiency, quality or cost is important may even require  ossification.
Such services, we claim, are therefore unsuitable for implementation on blockchains. 
E.g., a cloud service implemented on a blockchain will not be able to match the efficiency of centralized cloud providers, unless there is some ossification in the blockchain itself (e.g., in a layer-2 used by the service).
A caveat is such services can still be suitable for blockchains if the primary value comes not from the efficiency of the service but from transparency, security, and decentralization of the service.
That is to say, e.g., if a less efficient but transparent, secure, and decentralized cloud is valuable for certain use cases it may well be implemented on a blockchain. 

On the other hand, services where the human qualities of trust, creativity, relationship and  connection, common sense, intelligence, morality and ethics, empathy etc.---to name a few---are important can benefit from implementation on an unossified blockchain. 
E.g., a service where users can  query, interact, or obtain services from coalitions of experts around the world in a certain domain can benefit as a blockchain application.\footnote{The word expert is used loosely to mean a person or an organization that is skilled at providing a certain type of good or service.} 
As a simple example, we can have a service where patients interact with and obtain wholistic advice from a collaborative group of doctors from various specialties, rather than having to interact with one doctor at a time through time-consuming referrals.  
As before, such services can be implemented in an ossified manner as well. 
But an unossified service offers unique properties such as lower subjective bias, a higher trust and radical creativity, which are hard to achieve in an entrenched centralized organization. 
For instance, many Web-2 platforms exist today that allow end users to connect with and obtain services from domain experts.\footnote{Examples include Amazon, Fiverr, tele-health platforms etc.} 
But such interactions are often with individual experts or an  ossified team of experts, and seldom with a ``random sample'' of experts.  
This is because a centralized platform 
provides no incentives for experts to collaborate with a diverse set of other experts.  
We note that the type of services we claim as being suitable for blockchains are not necessarily new---many social network blockchains already exist today; trusted execution services have been the hallmark of blockchains since the beginning; many blockchains and dapps exist for purchasing creative digital media or art.   
However, even if the application is not new, our anti-ossification incentives can create a fundamentally different incentive structure for collaborating service providers and consequently a rich, human service experience for its end-users.

A secondary benefit of our proposed design, is that it naturally solves the scalability problem of blockchains. 
By collaboratively creating a block, a group of miners can process a much greater number of transactions compared to publishing the blocks solo. 
The design that emerges is different from existing scaling solutions like sharding, or DAG-consensus methods. 
In sharding too a group of miners (from different shards) process transactions in parallel to increase throughput.  
However, a miner is algorithmically assigned to a shard which creates a forced collaboration. 
E.g., a miner that has previously misbehaved (and, perhaps slashed as a result) can continue to be assigned to shards if it has sufficient stake. This creates a situation where honest miners are forced to collaborate with a known miscreant. 
In contrast, miners subjectively choose their collaborators in our proposed design. 
A known miscreant node is unlikely to be voluntarily chosen for a collaboration by an honest node. 
Changing identity by adopting a fresh public key and transferring stake does not help the miscreant. 
Our proposed method assigns an {\em importance} score to each public key. 
The importance is a measure of how much a node has collaborated with a diverse set of nodes in the past.  
Unlike tokens, importance is algorithmically earned over time and cannot be transferred between accounts easily. 
A miner can, among other factors, consider the importance of a node while choosing collaborators to avoid miscreants. 
Similar to sharding, in our proposed method a resource constrained node can verify the full blockchain only with the help of other collaborators or more capable nodes. 
It is also possible to reduce the confirmation latency of transactions. 
By letting nodes work in coalition groups, we can effectively reduce the number of votes necessary to confirm transactions, thus improving latency. 

In summary, we make the following contributions in this paper: 
\begin{enumerate}
\item We propose a fundamental redefinition of decentralization in blockchains as the extent of collaborative interaction that happens between diverse sets of users in the network,  regardless of how skewed resources are distributed across the users. 
\item Using block proposal as a concrete example, we provide an incentive mechanism that encourages decentralization. 
Under this scheme, we show that nodes do not have an incentive to form coalitions. 
\item We show how blockchain scalability can be improved naturally using our proposed framework. 
\item We provide discussions on extending our framework to other applications, such as decentralized autonomous organizations and smart contracts.  
\end{enumerate}

\section{Centralization in Blockchains}

Centralization is a property of the state of a blockchain system at a particular instant in time. 
It refers to the concentration of a resource essential for system operation---such as wealth, network infrastructure, or voting power---on the hands of a small number of users. 
A number of measurement studies have identified centralization occurring in multiple layers of the protocol stack of Bitcoin and other blockchains~\cite{sapirshtein2016optimal,gervais2014bitcoin,sai2021taxonomy,azouvi2018egalitarian,beikverdi2015trend,gencer2018decentralization,beikverdi2015trend}.  
These layers include, but are not limited to,  consensus~\cite{lewenberg2015bitcoin,beikverdi2015trend,10.1007/978-3-319-67816-0_18}, network~\cite{neudecker2018network,tapsell2018evaluation,roubini2018blockchain,feld2014analyzing,apostolaki2017hijacking}, wealth~\cite{chohan2022cryptocurrencies,fanti2019compounding}, governance~\cite{beck2018governance,hsieh2017internal,atzori2015blockchain}, geography~\cite{mao2023less,kim2018measuring}, exchanges~\cite{bohme2015bitcoin}, and equipment~\cite{8314142,8767923,raman2018dynamicdistributedstoragescaling,Ekblaw2016BitcoinAT}. 

The prevailing sentiments about centralization in blockchains are: (1) centralization is undesirable as it does not align with the core ethos of blockchains, (2) centralization eases the possibility of collusion attacks thereby weakening  security, and (3) centralization may not be completely avoidable.  
For example, in Kwon {\em et al.}~\cite{kwon2019impossibility}, the authors define a mathematical notion of decentralization based on how evenly resource is distributed across the users. 
They state that collusion attacks are more difficult if the resource power is more evenly distributed.
A key result of the paper is that a fully decentralized system is impossible when there is no ``Sybil cost'', i.e., when there is no extra cost for Sybil users to operate additional nodes.
The paper argues that public blockchains today do not have a Sybil cost, and therefore are susceptible to centralization. 
The problem of how to achieve a positive Sybil cost without relying on  a trusted-third party for identity management is left as an open question. 


To remedy the centralization problem, researchers have proposed many techniques such as using a non-linear function (e.g., square root of stake) for computing  voting power~\cite{motepalli2025decentralization}, ASIC-resistant hash functions in proof-of-work chains~\cite{cho2018asic}, quadratic voting in DAOs~\cite{dimitri2022quadratic}, and reward sharing schemes for fair formation of stake pools~\cite{brunjes2020reward}. 
In a way, all of these methods are techniques by which the symptoms of centralization can be managed. 

In this paper, we take a complementary approach to centralization prevention. 
We identify the core root-cause process by which centralization occurs, and propose mechanisms to stop the process. 
At any given time instant, a blockchain may be centralized due to multiple reasons: (1) the network could have strong centralization at genesis, which carries over through time; (2) entities external to the network (e.g., the government, law enforcement agencies, or a major ISP) may force certain users to have a disproportionately high (or, low) amount of resources and create centralization; (3) centralization can happen because economically it is the ``best'' strategy for the users. The first two reasons  are beyond our control. 
We focus on the third case.

\subsection{Centralization via Ossification}

In a free market, the objective of any seller is to provide a good or service of the highest quality at the cheapest possible price.\footnote{Technically the objective of a seller is to maximize profits, which can be done by improving product quality while reducing cost.}
Achieving a high product quality at a low price is challenging: it requires the manufacturing or the service process to be as {\em efficient} as possible which occurs when the seller invests money, time, and effort into optimizing the manufacturing or service pipeline.  
In many cases, providing competitive service is only possible through a coalition of sellers rather than by an  individual.\footnote{We use the word seller as a general term to mean members of the service providing coalition.} 
A coalition may be preferred if the service requires a diverse set of skills or expertise from different domains, which may be difficult for a single individual to have. 
A coalition may also be preferred to increase profits through economy of scale or reduce loss through risk pooling.  
As with individual sellers, to improve service efficiency in a coalition its members invest capital, time, and effort.\footnote{Capital can also come from external sources (e.g., investors). Similarly, service can involve effort from external bodies (e.g., contractors). We do not consider these as part of the coalition.}   
Today's markets are full of coalitions, in the form of partnerships, companies, corporations, and cooperatives. 

We identify three essential properties of coalitions in free markets (and in many mixed markets existing today). 
\begin{property} \label{property 1}
The coalition forms because a coalition is necessary to provide a competitive service in the market and/or due to economic factors. 
The coalition may upsize or downsize in reaction to market changes. 
\end{property}
\begin{property} \label{property 2}
Members of the coalition collaborate over time pooling effort and capital to increase service efficiency and product value.     
\end{property}
\begin{property} \label{property 3}
Members of the coalition rarely collaborate with competing coalitions in the market for providing service.     
\end{property}

This process by which coalitions form and are sustained represents a key centralizing force in today's economy.
Smaller coalitions can also be present internally within a larger coalition. 
E.g., a large company which is itself a coalition may internally consist of smaller coalitions in the form of divisions or departments.   
However, in this case the smaller coalitions arise not due to market competition but due to the internal organization and division of labor of the larger company.
Furthermore, there is typically at least some (if not, a lot of) collaboration between the smaller coalitions (divisions) in such a case. 
Therefore, Property~\ref{property 3} above is not true for these coalitions. 
We are primarily interested in coalitions that arise due to market forces, and which satisfy the Properties~\ref{property 1},~\ref{property 2}, and~\ref{property 3}. 
We use the term {\em ossification} to denote the process by which competing coalitions form in a market (Fig.~\ref{fig:ossification}).\footnote{The term is inspired by network ossification, which denotes the slow rigidification of the Internet architecture and protocols.} 
\begin{definition}
Ossification is the process by which a coalition forms and is sustained in a market while satisfying Properties~\ref{property 1},~\ref{property 2}, and~\ref{property 3}. 
\end{definition}

\begin{figure}[!tb]
    \centering
    \includegraphics[width=.9\textwidth,]{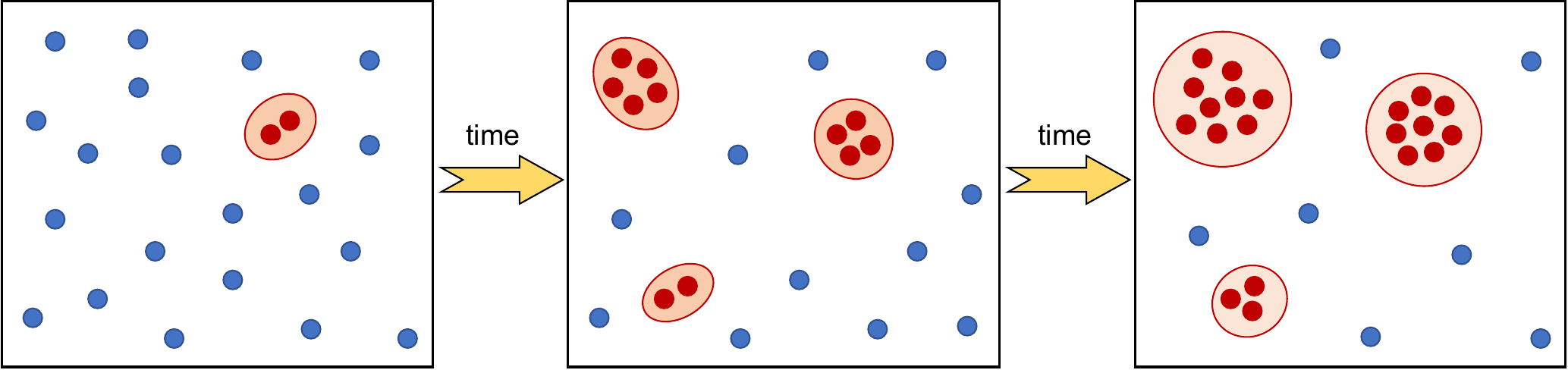}
    \caption{Sellers in a market have a natural tendency to ossify over time due to economic forces. Each blue dot represents a seller. The red clusters represent coalitions formed due to ossification.}
    \label{fig:ossification}
\end{figure}

The word ossification is justified because members of a coalition develop strong trust relationships with each other over time as they collaborate to compete in the market. 
Internal processes, communication pathways, overheads etc. are streamlined, and the coalition functions as a well-oiled machine. 
The strong collaboration among members over a long time leads to highly efficient products or services (extremely high quality at low cost), that is difficult to match by other smaller coalitions or coalitions that have not been around for as long. 
Many of the services and products that we enjoy today---smartphones, cloud services, healthcare facilities, supermarkets, airplanes etc.---are all the result of ossified coalitions. 
One could argue that ossification is, in fact, necessary to some extent to achieve the level of expertise needed to develop a complex service or product (such as the smartphone, or the cloud).  

Ossification is generally viewed as a favorable process, so long as it does not lead to monopolies and done in the interest of consumers. 
However, ossification has some drawbacks as well. 
As members of an ossified coalition actively collaborate to provide the best possible service, they typically do not collaborate with competitors in the same market.\footnote{Collaborations between competitors do occur on mutually beneficial issues. But even in those cases collaboration on the actual end-product on which the companies are competing is rare.} 
At first glance, this is almost a trivial statement. 
Why collaborate with a rival against  whom one is directly competing? 
From the point of view of the consumers, a collaboration between coalitions holding diverse viewpoints can potentially result in effective cross-pollination of ideas, lower bias, and produce creative breakthroughs that otherwise do not happen. 
In an ordinary market, the pressure on a coalition is to achieve the highest efficiency and earn the most amount of profits. 
This focus on a singular goal can cloud the coalition from exploring avenues that do not contribute to increasing efficiency in obvious ways. 
Coalitions that do focus on exploratory topics tend to be really affluent and constitute a minority. 

Not all products and services are a result of ossified coalitions. 
For example, in the movie industry (e.g., Hollywood), rarely do we see the same cast and crew collaborate over multiple projects. 
The creative foundations needed to make a movie are gained by encouraging collaborations between different sets of people for different movies. 

As in other markets, ossification happens in blockchains as well. 
The examples of centralization in blockchains cited at the beginning of this section---such as the formation of mining pools, or the centralization of voting power in DAOs---can be attributed directly or indirectly to ossification. 
It is, therefore, valuable to devise methods that discourage or eliminate ossification in blockchains. 

\section{Decentralization in Blockchains}

The word decentralization is often used in blockchains to mean the absence of centralization. 
Just as centralization is a function of time and can happen due to many reasons, it follows that decentralization is also a function of time and can happen due to any of those reasons. 
We argue that this is not a very useful definition for decentralization.
Consider a fully homogeneous proof-of-work blockchain network in which all nodes have the exact same amount of resources and uniform all-to-all network connectivity. 
Suppose a fraction (say, 30\%) of the nodes form a coalition and engage in selfish mining. 
Intuitively, the resulting network is not fully decentralized. 
However, by all resource measures the network appears to be fully decentralized. 
We conclude that an accurate definition of decentralization must consider not only the resource distribution across nodes, but also the intention of nodes to collaborate with other nodes in the network.  
We formalize this idea below. 

Consider a set of {\em free agents} in a market, where a free agent is an entity (a person, a group, or an organization) capable of making a decision of their own free will. 
We do not require the decisions that free agents take to be rational; but it is important that the decisions are taken by the free agents out of their own accord and not forced upon them (e.g., by a regulator, or by an algorithm in the case of blockchains). 
We assume agents are free to choose who they want to collaborate with---a property which we call as {\em freedom of choice}.  
Even if the agents can make decisions of their own, in many cases the agents don't have freedom of choice. 
E.g., geopolitical reasons can prevent trade between companies in certain regions of the world; or, in a sharded blockchain the specific shard that a node is assigned to (i.e., its collaborators) is algorithmically computed and outside of the control of the node. 
Our definition of decentralization is applicable for free agents that have freedom of choice. 
\begin{definition} \label{defn: decentralization}
A market with free agents having the freedom of choice is decentralized, if for any subset of free agents  there is a `significant' amount of collaboration between the subset and its complement. 
\end{definition}

\begin{figure}[!tbp]
    \begin{subfigure}[t]{0.45\textwidth}
        \centering
        \includegraphics[width=.8\textwidth]{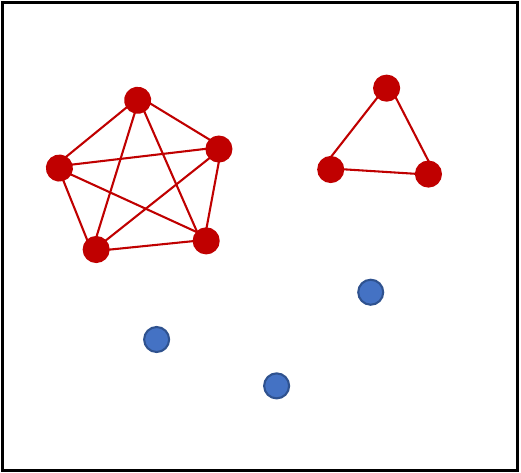}
        \caption{Ossified market.}
        \label{fig:collab1}
    \end{subfigure}
    \hfill
    \begin{subfigure}[t]{0.45\textwidth}
        \centering
        \includegraphics[width=.8\textwidth]{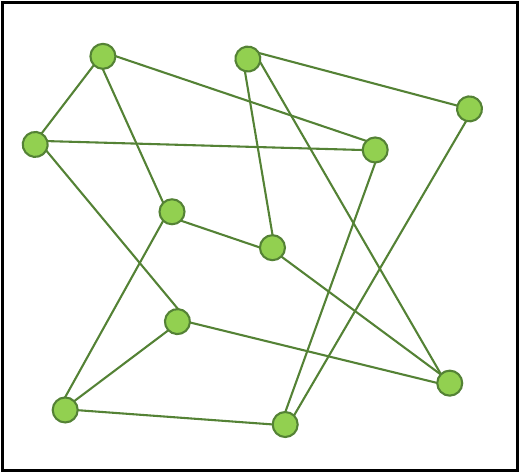}
        \caption{Decentralized market.}
        \label{fig:collab2}
    \end{subfigure}
    \caption{The collaboration graph in an ossified (left) and decentralized (right) market. Each dot is a seller, while the edges represent the presence of interaction between users. In (a), the red clusters are the ossified coalitions.}
    \label{fig:collab}
\end{figure}

The definition naturally forbids ossified coalitions where members collaborate only with other members of the coalition, and not with others. 
What `significant' amount of collaboration means is application dependent. 
We provide a more precise definition for the block proposal application in \S\ref{sec:blockunossify}. 
In short, we define decentralization not as the absence of centralization but as the presence of voluntary, value-adding interactions and collaborations between participants (Fig.~\ref{fig:collab}). 
There are some similarities and differences between the old definition of decentralization and Definition~\ref{defn: decentralization}. 



\smallskip 
\noindent
{\bf Similarities.} 
Just like the old definition, our new definition is also a continuous measure of decentralization. 
A system can be 0\% decentralized, 100\% decentralized, or fall anywhere in between. 
Our definition is not precise and can be implemented in multiple ways in practice. 
In a fully decentralized market, the ``collaboration graph''---computed as the pairs of free agents that have collaborated during the time horizon of interest---must look like an expander graph. 
If not, measuring to what extent the collaboration graph resembles an expander graph provides the extent of decentralization in the network. 
In certain cases, if the network is small, it is possible for each free agent to collaborate with all the other free agents 
within a reasonable amount of time. 
We do not place any constraints on the sparsity or maximum degree of the collaboration graph. 
Depending on the application, this constraint may be added. 

\smallskip
\noindent
{\bf Differences.}
In the old definition, the extent of decentralization can be measured at any given time instant. 
In our new definition, we can measure the extent of decentralization only over a time horizon (like ossification). 
Our definition is applicable only to systems that have human (or, organizations controlled by humans) participants. 
We cannot use the definition for a  network of machines that don't have free will (e.g., a private blockchain managed by a single entity).  
Our definition is complementary to the old definition. 
For instance, we can consider a network where all the nodes are highly collaborative, but the resource distribution is highly non-uniform. 
Such a network is centralized per the old definition, but decentralized according to our definition. 
It is possible to develop a hybrid definition that considers both resource distribution and collaboration. 
However, the utility of the old definition is limited to just quantitatively measuring how resources are concentrated in the network. 
The definition does not naturally suggest any mechanisms by which the resource concentration can be mitigated. 
Definition~\ref{defn: decentralization}, on the other hand, is more general in the sense that it not only provides a measure to quantify decentralization but also presents a natural method to  prevent centralization from happening in the first place via the typical process of ossification (as we will see in the coming sections).\footnote{Atypical processes of centralization cannot be avoided even by our scheme.} 
The old definition is a worst-case measure. 
If resources are concentrated in the hands of a few, those few {\em could} attack the system, not that they would.
Since we cannot predict the behavior we err on the side of caution, and require all resources to be uniformly distributed. 
Our definition is a typical-case measure. 
Even if resources are concentrated at the hands of a few, if those few are willing to collaborate with everyone in the network, we feel there is no need for worry. 
In the worst case, collaborations can abruptly stop and the nodes in power can turn evil. 
But in the typical case, so long as the nodes do not have an economic incentive to turn evil, we believe they will not turn evil. 

\smallskip
\noindent
{\bf Discussion.}
An important aspect in Definition~\ref{defn: decentralization} is that the collaborations have to be made voluntarily by the nodes, and not forced by external factors. 
In other words, we observe when given a choice who the nodes choose to work with. 
If a certain group of nodes form a clique and choose to work only with other members of the group, we call the network as centralized. 
If each node voluntarily collaborates with all other nodes, we call the network as decentralized. 

The idea of encouraging greater collaboration to combat monopolies, encourage creativity, reduce bias, or to enforce a desired shared culture has been voiced by many experts from various disciplines in the past~\cite{demisinterview,doudna2020viable,vitaliksanctuary,torvalds2001just,farmer2004pathologies,berners1999weaving,ninessop}. 
The focus of this paper is to formalize this idea into a mathematical framework, which can then be used for decentralizing blockchains.  

In the following, the word decentralization means decentralization as per Definition~\ref{defn: decentralization}. 
Where there is ambiguity, we call our definition as ``collaborative'' decentralization and the old definition as ``resource'' decentralization. 

\section{System Model and Problem Statement}

The ideas we have presented thus far are general and can be applied to any layer of the blockchain protocol stack (we hypothesize there are applications beyond blockchains as well, where this could be applied).
However, for concreteness we first focus on a key application area: the consensus protocol. 

\smallskip
\noindent
{\bf System model.} We consider a blockchain system comprising of a set of nodes $V$, with $n = |V|$. 
Assuming a proof-of-stake (PoS) system for concreteness, 
each node $v \in V$ has a stake of $s_v \geq 0$ with $\sum_{v \in V} s_v = 1$. 
Time proceeds in discrete rounds, $t=0,1,\ldots$, with one block published each round. 
Let $V[t] \in V$ be the proposer assigned to be the block publisher for round $t$. 
The probability that $V[t] = v$ for a node $v \in V$ at time $t$ is $s_v$. 
$R[t]$ is the block proposer reward that the proposer $V[t]$ gets at round $t$ for publishing the block.
In addition to the block proposer reward, the proposer also receives transaction fees as a reward.
We do not model transaction fee rewards for the time being. 
A fraction of the nodes with a total stake of up to $1/3$ may be malicious in the system. 
Each node has a public-private key pair. 
Digital signatures are secure. 

We assume the block size can be arbitrarily increased in size to accommodate any additional metadata required by a proposed algorithm, without hurting the block dissemination time. 
We also assume nodes have unlimited compute capacity. 
In later sections, we discuss the actual overhead introduced by our proposed method and discuss possible optimizations for practice. 

\smallskip 
\noindent
{\bf Problem statement.}
We seek to design a block reward computation mechanism for $R[t]$ such that there is no incentive for ossification. 
That is, for any subset $S \subset V$, if the nodes in $V$ pool their resources (stake),   it must result in a lower reward for each node compared to following the  protocol.  
In today's blockchains, the expected reward for a node is identical whether or not it joins a stake pool. 
It is only the variance in reward that decreases upon joining a stake pool. 
What we require in our problem is the expected reward for a node should be strictly lower if it joins a stake pool, compared to not joining a pool (more precisely, ossifying with a pool). 

\section{Block Production and Reward Mechanism} \label{sec:blockunossify}

\subsection{Overview}
We present a method by which any PoS consensus protocol can be extended to allow a block to be jointly proposed by a subset of nodes, instead of a single node.
Our extension requires a real-valued state for each node, a global state capturing the collaboration patterns between the nodes, a new field in the block header, and a modification of the block reward computation mechanism. 
It does not alter the core consensus protocol or diminish its security. 

Before a block is published, the block publisher invites other nodes of its choice to collaborate on the block publication. 
When the block is published, the signatures of all the collaborators are included within the block header in addition to the block proposer's signature. 

We introduce a new real-valued state associated with each node called {\em importance}, that measures the amount of collaboration a node has done over time. 
When a node collaborates with a publisher, a part of the node's importance goes to the publisher when the block is published. 
Unlike payment tokens, importance can be transferred only via collaboration and that too at a slow, fixed rate. 

A constant-rate ``tax'' on the nodes decays their importance slowly.  
A node can increase its importance when it is a block publisher by collaborating with other nodes.
A node can also increase its importance through a ``tax refund'' if it is well-behaved and has suffered a loss in importance due to collaboration. 
The amount of importance received from the tax refund depends on the importance of the nodes it collaborates with, and the overall collaboration graph of the network.  

The publisher of the block receives a reward that is proportional to the total importance of all the collaborators (including the publisher) in the block. 
The publisher may choose to allocate a portion of this reward to its collaborators. 
If the block produced is incorrect, only the publisher is slashed. 
The collaborators do not get slashed. 
We explain these ideas in more detail in the following. 
Alternative solutions are possible by considering variations of these ideas. 

\subsection{Importance}

Importance is a real-valued state associated with each node that tracks the extent to which the node has collaborated with other nodes. 
For a node $v \in V$, we let $I_v[t] \geq 0$ be the importance of $v$ at the beginning of round $t$. 
Importance is a conserved quantity. It can neither be created nor destroyed. 
It can only be transferred to or received from another node. 

\smallskip
\noindent
{\bf The collector node.}
In addition to the nodes $V$, we introduce a hypothetical node $c$ which we call the collector. 
The collector node has its own associated importance $I_c[t] \geq 0$ at the beginning of round $t$. 
The collector node collects a tax in importance from all the nodes at each round. 
For a system parameter $\beta \in (0, 1)$, the tax collected from any node $v \in V$ at round $t$ is $\beta I_v[t]$. 
The collector also gives importance back to the network in the form of a tax refund awarded to certain well-behaved nodes. 
The exact amount awarded and who it is awarded to depends on the collaboration patterns of the nodes, and will be discussed later. 
If a node is idle, i.e., it is not a frequent publisher (due to low stake) and it is not a frequent collaborator, the node will eventually lose all of its importance to the collector.
If a group of nodes form an ossified coalition, with each node in the group only collaborating with other nodes in the group, the aggregate importance of the entire group goes to zero. 
We set the total amount of importance in the network (including the collector) to 2.
So, we have the invariance $\sum_{v \in V} I_v[t] + I_c[t] = 2$ for all time $t \geq 0$. 
We also set $I_c[0] = 1$. 

\subsection{Collaboration DAG}

Before a block at round $t$ is published by $V[t]$, the publisher negotiates with other nodes of its choice to form a collaboration set $S[t] \subseteq V$ with $V[t]\in S[t]$. 
The reason it is a negotiation will be clear shortly. 
The collaboration set is organized as a directed acylic graph (DAG) $D[t]$ with the vertices being the nodes in $S[t]$ and $V[t]$ as the root. 
This DAG $D[t]$ is published in the block header at time $t$, including signatures of the members in $S[t]$ attesting to the DAG structure. 
The block reward $R[t]$ earned by the publisher $V[t]$ is set as
\begin{align}
R[t] = \sum_{v \in D[t]} I_v[t](1-\beta),
\end{align}
i.e., the block reward equals the total importance of the collaboration set after paying tax to the collector. 
The block reward is awarded only to the publisher $V[t]$. 
Each non-root node in the DAG $D[t]$ pays an importance fee to all of its predecessors in the DAG, and receives an importance free from its successors in the DAG.  
For each $v \in D[t]$, let $I'_v[t]$ be the importance of $v$ after sending importance to its predecessors and receiving importance from its successors.
We design the importance transfer mechanism such that it satisfies the following properties. 
\begin{property}[Conservation of importance] \label{prop:convimp}
The total importance of the nodes in the DAG before and after the transfer is the same, i.e., $\sum_{v \in D[t]} I'_v[t] = \sum_{v \in D[t]} I_v[t](1-\beta)$. 
\end{property}
Therefore, a collaboration among Sybil nodes---regardless of the DAG structure---cannot increase the total amount of importance of the Sybils. 
\begin{property}[Slow transfer of importance] \label{prop:slowtransf}
For a DAG with a total stake of $s$ and a total importance of $i$, the maximum amount of importance a node can receive in the DAG is $\alpha i s / \epsilon$,  
where $\epsilon$ is the smallest denomination of the blockchain's native token and $\alpha \in (0,1)$ is a system parameter governing importance transfer within the DAG (see Appendix~\ref{apx:imptransfDAG}). 
\end{property}
Importance captures the {\em long term} collaboration behavior of a node and is, therefore, slow varying. 
This allows nodes to form stable collaboration relationships with peers. 
From a security standpoint, a slow varying importance increases the cost of certain attacks for the Sybils. 
E.g., if a Sybil node forges an incorrect signature during collaboration causing the publisher to get slashed, it cannot immediately assume a new identity with the same stake and importance as before. 
\begin{property}[Fairness] \label{prop:fairness}
Let $S \subset D[t]$ be any path-closed subset of nodes containing $V[t]$, i.e., for any $u, v \in S$ and for any path from $u$ to $v$ in $D[t]$, the path is contained in $S$. 
We have
\begin{align}
\sum_{v \in S} (I'_v[t] - I_v[t](1-\beta))  &\geq 0 \geq \sum_{v \in D[t]\backslash S} (I'_v[t] - I_v[t](1-\beta)). 
\end{align}
For any node $v \in D[t]$, the final importance value of $v$ is lower bounded as 
\begin{align}
I'_v[t] &\geq I_v[t](1-\beta)(1 - \frac{\alpha}{\epsilon}\sum_{u \in Q(v)} s_u).  
\label{eq:fariness2} 
\end{align}
\end{property}
The closer a node is to the root of the DAG, the more importance it receives from the nodes below it. 
Nodes close to the sink(s) of the DAG lose importance to the nodes above them. 
The root of the DAG receives the most importance. 
Such a mechanism not only provides the root with an incentive to gather a large collaboration set, but it also gives each node of the DAG an incentive to find more collaborators of their own.
Thus, the collaboration set itself can be collaboratively (and hierarchically) built, which reduces overhead on the root and allows for rapid construction of large collaboration sets. 
For many applications, it may suffice to just have a root and one level of child nodes below it. 
We use a DAG as it is general and can support a flexible range of applications. 
\begin{property}[Sybil resistance] \label{prop:sybil}
Consider any node $v \in D[t]$. 
Suppose $v$ is replaced by any DAG $A[t]$ with total stake $s_v$ and total importance $I_v[t](1-\beta)$.
The total fraction of importance received by $A[t]$ exceeds the fraction of importance received by $v$ in the original DAG by a factor that is at most 
\begin{align}
    1 + \frac{\frac{\alpha^2}{2}(\frac{s_v^2}{\epsilon^2} - \frac{s_v}{\epsilon}) + O(\alpha^3)}{(1-(1-\alpha)^\frac{s_v}{\epsilon}))},      
\end{align}
which goes to 1 as $\alpha \rightarrow 0$. 
\end{property}
Here `replaced by a DAG' means the parents of $v$ in the original DAG have edges to all the roots of the new DAG;  
and the sinks of the new DAG have edges to all the children of $v$ in the original DAG. 
Due to this property, a Sybil user cannot hope to game the protocol by pretending to be distinct users. 
There is also a social aspect to Sybil resistance. 
A single adversary may create multiple Sybil personalities hoping the victim trusts at least one or a few of these personalities. 
We do not model this attack or behavior. 
An honest entity can also run multiple nodes. 
But, the honest node would be publicly transparent about its ownership of those nodes. 

\smallskip
\noindent
{\bf Importance transfer mechanism.} 
We illustrate the importance transfer mechanism on directed trees, as it is easier to understand. 
The general case of importance transfer on DAGs is discussed in Appendix~\ref{apx:imptransfDAG} (Fig.~\ref{fig:dagex}). 

\begin{figure}[!tbp]
    \begin{subfigure}[t]{0.45\textwidth}
        \centering
        \includegraphics[width=.8\textwidth]{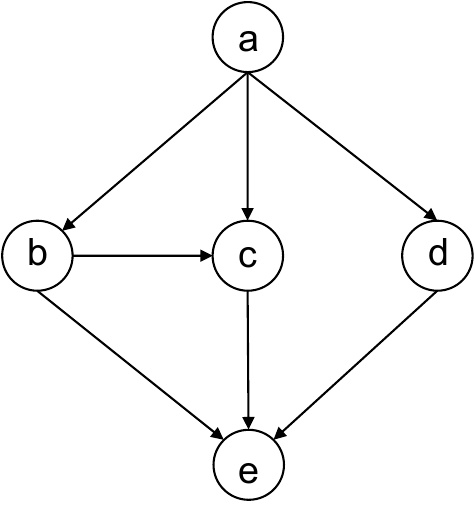}
        \caption{}
        \label{fig:dagex1}
    \end{subfigure}
    \hfill
    \begin{subfigure}[t]{0.45\textwidth}
        \centering
        \includegraphics[width=.8\textwidth]{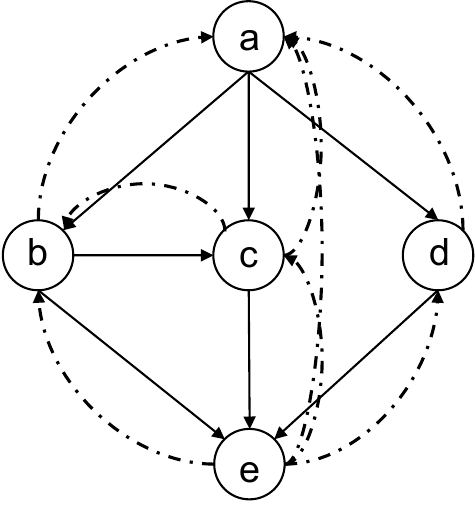}
        \caption{}
        \label{fig:dagex2}
    \end{subfigure}
    \caption{(a) Example of a collaboration DAG. (b) Each node in the DAG transfers a fraction of its importance to its ancestors in the DAG.}
    \label{fig:dagex}
\end{figure}

Consider a directed, rooted tree $D[t]$ with the publisher $V[t]$ as the root. 
Consider any node $v \in D[t]$. 
The importance of $v$ after taxation is $I_v[t](1-\beta)$. 
Let $V[t], v_1, v_2, \ldots, v_k, v$ be the path to reach $v$ from the root. 
Node $v$ first transfers a fraction $(1-(1-\alpha)^{s_{V[t]}/\epsilon})$ of its available importance to $V[t]$. 
Out of the remaining importance, it transfers a fraction $(1-(1-\alpha)^{s_{v_1}/\epsilon})$ to $v_1$. 
The process repeats until $v$ has transferred importance to all the nodes on the path $V[t], v_1, \ldots, v_k$. 
If $\alpha$ is small relative to the stake values, the quantity $(1-(1-\alpha)^{s/\epsilon})$ can be approximated as $\alpha s/\epsilon$. 
The above process can, therefore, be approximated as node $v$ sending an importance of $I_{v}[t](1-\alpha) \alpha s_{V[t]}/\epsilon$ to node $V[t]$, an importance of $I_{v}[t](1-\alpha) \alpha s_{v_1}/\epsilon$ to node $v_1$, $I_v[t](1-\alpha) \alpha s_{v_2}/\epsilon$ to node $v_2$, and so on, till node $v_k$. Each node in the tree transfers importance this way to all of its predecessors. 
The root node only receives importance. 
In practice, for a small value of $\alpha$ we can use $\beta = 100 \alpha / \epsilon$.  
\begin{theorem} \label{thm:importancetransfer}
The importance transfer scheme satisfies Properties~\ref{prop:convimp},~\ref{prop:slowtransf},~\ref{prop:fairness}, and~\ref{prop:sybil}.  
\end{theorem}
(Proof in Appendix~\ref{apx:proofofthm1}). 

Our importance transfer is identical to an asymmetric Shapley value with the permutation weights chosen to provide Sybil resistance. 
For any subset $S \subseteq D[t]$, define the utility $u(S)$ of $S$ as the total importance of all the nodes reachable from $S$ (including the total importance of $S$ itself), i.e., 
\begin{align}
u(S) = \sum_{\substack{v \in D[t]: \exists u \in S \text{ such that } \\ v \text{ is reachable from } u}} I_v[t](1-\beta).     
\end{align}
The utility function is a monotone, submodular function and has an associated polymatroid.
For any permutation $(v_{\sigma(1)}, v_{\sigma(2)}, \ldots, v_{\sigma(|D[t]|)})$ of the nodes in $D[t]$, where $\sigma$ is the permutation function, the greedy solution of
\begin{align}  x_{v_{\sigma(i)}} = u(\{ v_{\sigma(1)}, v_{\sigma(2)}, \ldots, v_{\sigma(i)} \}) - u(\{ v_{\sigma(1)}, v_{\sigma(2)}, \ldots, v_{\sigma(i-1)} \}), 
\end{align}
is a corner point on the base of the polymatroid. 
The standard symmetric Shapley value, in addition to being computationally expensive, does not satisfy Property~\ref{prop:sybil}. 
We, therefore, consider a non-standard asymmetric Shapley value by choosing a subset of permutation to average over instead of all possible permutations. 
Details are provided in Appendix~\ref{apx:imptransfDAG}. 

\smallskip
\noindent
{\bf Incentives.} To maximize the block rewards, the publisher has an incentive to form and use a large collaboration set (large, in the sense of importance). 
The publisher may give a portion of the the block reward to its collaborators.
E.g., for a parameter $\gamma  \in (0, 1)$ the publisher keeps $I_{V[t]}(1-\beta ) + \sum_{v\in D[t]: v\neq V[t]} I_v[t](1-\beta)\gamma$ of the reward, and each node $v \in D[t], v \neq V[t]$ receives $I_{v}[t](1-\beta)(1-\gamma)$ of reward.
Other reward allocations are possible. 
The parameters (e.g., $\gamma$) of the reward allocation can be  determined through a private negotiation between the publisher and the collaborators. 
During the negotiations, the publisher can also decide the exact structure of the DAG and which collaborator goes where in the DAG. 
It is completely up to the publisher to invite a node to be a collaborator. 
It is also up to the invited node whether to accept the invitation or not. 

A node loses a fraction $\beta$ of its importance each round to taxes, regardless of whether it is part of the collaboration set or not. 
If the node rejects an invitation to collaborate, it loses its tax and gains no monetary reward. 
If the node accepts the invitation, depending on its location in the DAG it may overall lose or gain importance. 
However, the node can receive monetary compensation from the publisher. 
Thus, there is an incentive for the node to join a collaboration as it can benefit the node by increasing its importance and/or providing a monetary reward. 
Who a node chooses to collaborate with has important implications on the overall rewards earned by the node. 
E.g., if the node is going to be a publisher soon, the node may not want to spend its importance in the current round; 
or, if the node wants to maximize its tax refund (\S\ref{s: tax refund}), 
But in all cases, it is detrimental to the node to not collaborate at all, or to collaborate only with a fixed, ossified coalition. 

\smallskip
\noindent
{\bf Making collaborations useful.}
In our discussion so far, the value of collaboratively producing blocks is to prevent ossification. 
However, collaborators can also be tasked with doing additional useful work during the block production process such as executing transactions. 
We discuss how this idea can be used to scale the blockchain in \S\ref{s:scalingperf}. 

\subsection{Tax Refund} \label{s: tax refund}

The collector node collects a tax each round for a few reasons: (1) it removes importance out of idle or unused nodes, (2) it forces nodes to join a collaboration, (3) it allows for the importance of an ossified coalition to go to zero, and (4) it can be used to enforce an all-to-all collaboration paradigm. 
We discuss the fourth point here. 

Tax collected by the collector is transferred back to well-behaving nodes in the form of a tax refund (otherwise, the total importance of all the nodes will go to zero).
A key challenge here is in identifying which nodes are well behaved, and how much to offer in refunds to those nodes. 
To solve this, we consider an interaction graph $G[t]$ for each round $t$. 
Each node $v \in V$ is a vertex in $G[t]$. 
At genesis, the graph $G[0]$ does not have any edges. 
At time $t$, the graph $G[t]$ can be derived from $G[t-1]$ as follows. 
For any two nodes $u, v \in V$, let $I_{(u,v)}[t]$ be the amount of importance transferred from $u$ to $v$ at round $t$. 
We set $I_{(u,v)}[t] = 0$ if $u \notin D[t]$ or $v \notin D[t]$. 
For any edge $(u, v) \in G[t]$ and for any time $t$, let $w_{(u,v)}[t]$ be the weight of the edge $(u,v)$ at time $t$ in $G[t]$. 
The weight of all the edges is set to zero at genesis. 
Then, 
\begin{align}
w_{(u,v)}[t] = (1-\alpha) w_{(u,v)}[t-1] + \alpha I_{(u,v)}[t],
\end{align}
for all $u,v \in D[t]$ and for all time $t$. 
The edge weight $w_{(u,v)}[t]$ is an exponentially-weighted moving average of the importance transferred from $u$ to $v$.  

\begin{figure}[!tb]
    \centering
    \includegraphics[width=.9\textwidth,]{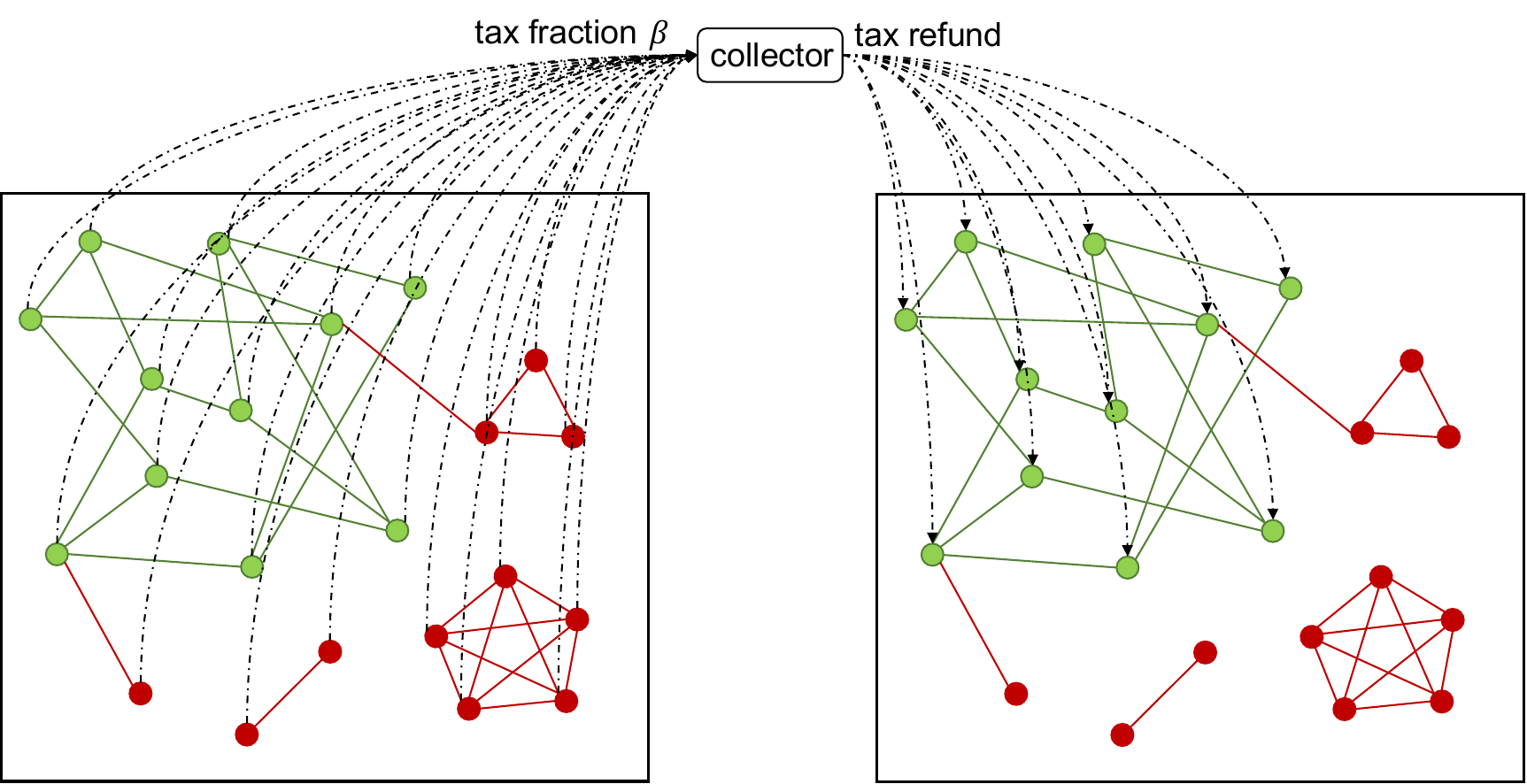}
    \caption{A $\beta$ fraction of importance is collected as tax from all the nodes at the beginning of a round. Nodes that are part of the maximal expander receive a tax refund at the end of the round, provided the total stake of the expander is at least 0.5.}
    \label{fig:tax}
\end{figure}

For each node, the interaction graph measures how much importance the node has sent and received through collaborations in roughly the last $1/\alpha$ rounds. 
Once we have $G[t]$, we derive a maximal subgraph $G'[t]$ from $G[t]$ that is an $\phi$-expander for a decentralization  parameter $\phi$. 
Since $G'[t]$ is a $\phi$-expander, for any subset $S \subset G'[t]$, we have 
\begin{align}
 \frac{\min(\sum_{(u, v) \in G[t]: u \in S, v \in V\backslash S} w_{(u,v)}[t], \sum_{(u, v) \in G[t]: u \in V\backslash S, v \in S} w_{(u,v)}[t])}{\min(\sum_{u \in S} I'_u[t], \sum_{u \in V\backslash S} I'_u[t])} \geq \phi . 
\end{align}
Recent works have proposed efficient near-linear time randomized expander decomposition algorithms~\cite{fleischmann2025improved,saranurak2019expander,sulser2024near,arvestad2022near}. 
If the total stake in $G'[t]$ is less than 50\%, none of the nodes receive a tax refund at round $t$. 
If the total stake of $G'[t]$ is greater than 50\%, each node 
$v \in G'[t]$ receives an importance of $ I'_v[t] (\sum_{v \in G'[t]} s_v) \alpha (I_c[t] + (2 - I_c[t])\beta)/(\sum_{v \in G'[t]} I'_v[t])$ from the collector, where $(I_c[t] + (2 - I_c[t])\beta)$ denotes the importance of the collector after taxing the nodes at round $t$.
Tax refunds are issued at the end of a round (Fig.~\ref{fig:tax}). 

\smallskip
\noindent
{\bf Rationale.}
The expander graph $G'[t]$ (if it exists) ensures that any subset of nodes $S \subset G'[t]$ with $\sum_{v\in S} I'_v[t] \leq 0.5 \sum_{v \in G'[t]}I'_v[t]$ has outgoing edges of weight of at least $\phi \sum_{v \in S} I'_v[t]$. 
Thus, a set of malicious nodes cannot refuse to collaborate with honest nodes, without losing their importance refunds. 

\subsection{Analysis}

Consider an ossified set of nodes $S \subset V$. 
We say a collaboration between $S$ and $V \backslash S$ is $ \phi_\mathrm{out}$-{\em weak} if $\sum_{(u,v): u\in S, v\in V\backslash S}w_{(u,v)}[t] \leq \phi_\mathrm{out} \sum_{v\in S}I'_v[t]$ for all time $t$. 


\begin{theorem}[Preventing ossification]
Consider an ossified set of nodes $S \subset V$ having an $, \phi_\mathrm{out}$-weak collaboration with $V \backslash S$ with $\phi_\mathrm{out} < \phi$. 
Then, the total importance of the set $S$ is bounded as 
\begin{align}
  \limsup_{t\rightarrow \infty}\sum_{u \in S} I_u[t] \leq \frac{2\alpha \sum_{u\in S}s_u}{\beta(\epsilon - \alpha\sum_{u\in S}s_u) + \alpha \sum_{u\in S}s_u}.   
\end{align}
\label{thm:preventingossification}
\end{theorem}
(Proof in Appendix~\ref{s:proofofthmprevossification}). 

By choosing a tax rate $\beta$ that is significantly larger than $\alpha$ (e.g., $\beta = 100\alpha/\epsilon$), we can ensure that coalitions that do not strongly collaborate with others outside the group achieve a small eventual importance. 
By a similar argument as above, we can show that if a set $S$ has no collaborations with $V \backslash S$, the importance of $S$ goes to zero. 
\begin{corollary}
For an ossified set of nodes $S \subset V$ with $\sum_{(u,v): u\in V\backslash S, v \in S}w_{(u,v)}[t] = 0$, we have $\limsup_{t\rightarrow\infty}\sum_{u\in S}I_u[t] = 0$. 
\end{corollary}

We say a subset $S \subset V$ achieves a {\em strong} collaboration with itself if the total stake of $S$ is at least 0.5, i.e., $\sum_{u\in S}s_u >  0.5$, and 
\begin{align}
 \frac{\min(\sum_{(u, v) \in G[t]: u \in S', v \in S\backslash S'} w_{(u,v)}[t], \sum_{(u, v) \in G[t]: u \in S\backslash S', v \in S'} w_{(u,v)}[t])}{\min(\sum_{u \in S'} I'_u[t], \sum_{u \in S\backslash S'} I'_u[t])} \geq \phi, 
\end{align}
for all $S' \subset S$ and for all time $t$. 
If $S$ has a strong collaboration internally but has a weak collaboration with $V \backslash S$, then $G'[t] = S$ for all $t$ and the total amount of importance in $S$ grows to a large value.
We first show the following Lemma. 
\begin{lemma} \label{lem:iclowerbound}
The importance of the collector is at least 1, i.e., $I_c[t] \geq 1 ~\forall t > 0$. 
\end{lemma}
(Proof in Appendix~\ref{s:proofoflemiclowerbound}). 

Let $w_v[t] = (1-\alpha) w_v[t-1] + \alpha I_v[t-1] = \sum_{t'=0}^t \alpha (1-\alpha)^{t-t'}I_v[t]$. 
That is, $w_v[t]$ captures the average importance of $v$ over the past roughly $1/\alpha$ rounds. 
We provide a lower bound for the average importance of $S$ in the following. 
\begin{theorem}[Rewarding  decentralization]
Consider a set of nodes $S \subset V$. 
If (1) $S$ has a strong collaboration with itself, (2) $S$ has a  $\phi_\mathrm{out}$-weak collaboration with $S'$ for any $S'\subseteq V \backslash S$, and (3) $S'$ has a $\phi_\mathrm{out}$-weak collaboration with $S$ for any $S'\subseteq V\backslash S$, where $\phi_\mathrm{out} < \phi$ then
\begin{align}
\liminf_{t\rightarrow\infty} \sum_{v\in S} w_v[t] \geq (\sum_{v\in S}s_v) \frac{(2 - 4\beta + 6\beta^2 -2\beta^3)}{2-\beta} - \frac{\phi_\mathrm{out}}{\beta}.     
\end{align}
\label{thm:rewardingdecentralization}
\end{theorem}
(Proof in Appendix~\ref{s:proofofthemrewardingdecent}).

\subsection{Attacks}
\smallskip
\noindent
{\bf Ossification attack by inclusion.}
A set of malicious nodes $V\backslash S$ can collaborate only with a small, targetted set of victim nodes $S' \subset S$ without collaborating with any node from $S \backslash S'$. 
The idea here being if we consider the overall set $S' \cup (V\backslash S)$, it has a poor conductance in $G[t]$.
However, if the set $S'$ has a strong conductance in $G'[t](S)$ (i.e., $S'$ is well connected with $S\backslash S'$), it will still be included in the maximal expander subgraph $G'[t]$. 
The nodes in $V\backslash S$ will be excluded from $G'[t]$. 
This attack, therefore, would not work. 

\smallskip
\noindent
{\bf Ossification attack by exclusion.}
Another attack is the nodes in $V\backslash S$ can collaborate only with nodes in $S \backslash S'$, and not collaborate with any node from $S'$ for a target set $S' \subset S$. 
In this case, even if $G[t](S)$ and $G[t]((S\backslash S') \cup (V\backslash S))$ are expanders, the overall graph $G[t]$ may not be an expander. 
Therfore, it is possible that some nodes in $S'$ are excluded from $G'[t]$ (some or all nodes in $V\backslash S$ may be included in $G'[t]$). 
Suppose $V\backslash S$ has a total stake of $33\%$. 
For this attack to work, the target set $S'$ must be sufficiently large. 
E.g., if $S'$ contains just 1\% of the stake, the overall graph $G[t]$ might still be an expander. 
However, if $S'$ includes, say, $33\%$ of stake then it is possible that some nodes of $S'$ are excluded from $G'[t]$. 
We believe the solution to this attack is sociological, rather than algorithmic. 
When a collaboration rift between two (or, more) factions is detected in the network, it is in the interest of the neutral set of nodes $(S\backslash S')$ to ensure that the stake of the maximal expander in $G[t]$ does not drop below $50\%$. 
Moreover, the presence of non-collaborating factions increases the centralization of the network which ultimately diminishes the value of the blockchain to end users. 
Keeping these considerations in mind, the neutral set of nodes can decide on an appropriate collaboration strategy for the network, such as (1) encouraging the factions to collaborate with each other, or be isolated, (2) choose one of the factions and increase collaborations with them while reducing collaborations with the other. 



\smallskip
\noindent
{\bf Betrayal attack.}
An attacker can pretend to be a trustworthy node(s) and collaborate with a large fraction of the honest nodes. 
Eventually, the attacker deliberately includes an illegal operation in a block of an honest publisher (e.g., invalid signature) with whom the attacker is collaborating and causes the publisher to get slashed. 
The attacker than attempts to create a new public key identity, and transfers all of its stake to it. 
The attacker then tries to rejoin the network as a new person. 
However, unlike stake, transferring importance (1) takes time, and (2) cannot be anonymized (i.e., no anonymous mixers are available for importance transfer).  
Therefore, the network can easily know the new identities the attacker is trying to take and blacklist them too.


\smallskip
\noindent
{\bf Edge weight amplification attack.}
The weight $w_{(u,v)}[t]$ of an edge $(u,v) \in G[t]$ has memory of past importance transfers from $u$ to $v$. 
Suppose the nodes $v, v_1, v_2, \ldots, v_k$ are all controlled by the attacker. 
From $v$, the attacker can transfer $i$ amount of importance to $v_1$ which can then send the importance back to $v$. 
Node $v$ can then send this importance to $v_2$ which, once again, can send the importance back to $v$, and so on. 
If the memory of edge weights is long, it would appear as if node $v$ has transferred a net of $ki$ importance to other nodes, when, in fact, it is the same pot of importance that is being sent back and forth. 
If honest nodes do not engage in this behavior, the edge weights of the outgoing edges from the attacker nodes can be severely inflated in value giving the attackers an unfair advantage.
However, in our proposed solution the rate at which the edge weight memory decays (i.e., $\alpha$) is exactly the same as the rate at which a node can transfer its importance to another node. 
Therefore, by the time it takes for node $v$ to substantially transfer its importance to node $v_1$, the memory of the previous transfer that node $v$ made of the same importance decays down to zero in the edge weight. 
We leave a rigorous analysis of these attacks to future work. 

\smallskip
\noindent
{\bf Namesake collaboration.}
Nodes can form coalitions (pools) with each pool having a leader. 
Nodes retain their stake and importance and do not transfer them to their pool leader. 
The pool leaders negotiate collaboration agreements with each other on behalf of their pool members. 
Pool leaders may also do the bulk of the actual collaboration work (e.g., transaction execution), requesting their members only for a signature at the end. 
Members transfer a portion of their reward to the pool leader. 
In this case, the true trust relationships exist only between the pool leaders, and between a pool leader and its pool members, even if the collaboration graph $G[t]$ looks like an expander.
However, a pool member that is capable of doing the computational work by itself can bypass its pool leader and engage in the collaboration by itself, thus avoiding the pool leader fees.  
To prevent ``lazy'' pool members from delegating the collaboration process to their pool leaders, a stronger collaboration verification system beyond just signatures may be necessary.

\section{Scaling Blockchain Performance} 
\label{s:scalingperf}

Collaborations in a decentralized 
blockchain can be used to scale the performance of the blockchain. 
In the following we outline a method for scaling throughput; as before, other solutions are possible. 
Our method is inspired by blockchain sharding, but has subtle differences. 

We partition the public key address space into $k$ {\em tracks}, for a global parameter $k \in \mathbb{N}, k > 0$. 
The parameter $k$ can be increased or decreased as needed by the network. 
A node in the network, depending on its capacity, stores the states and processes transactions from only a small number (e.g., 1) of tracks. 
A node is free to choose which track(s) it desires to operate on. 
There is a single global blockchain (unlike sharding), with blocks published by proposers following a PoS protocol. 
A block proposer collaborates with other nodes operating in different tracks to collaboratively build a block.
Each block has a sequence of transactions, where a transaction can be from any track.  
Note that publishers are not required to collaborate or include transactions from all tracks, though it is in their interest to do them.  
Cross-track transactions are divided into multiple single track transactions (e.g., with the help of bridge addresses). 
All component transactions of a cross-track transaction are included and executed within the same block. 
Nodes operating in a track are connected to other nodes in the track through a p2p network.
After a block is prepared, a collaborator node collects transactions in the block that belong to its track and broadcasts them in its network. 
The block proposer is slashed if the block contains an invalid transaction or execution. 
If a node does not have the capacity to store data from all tracks, it can verify the correctness of a block only with the help of collaborating nodes from other tracks. 
Depending on the number of tracks and size of collaborations desired, the throughput of the blockchain can be significantly increased. 
Our solution can be likened to a multi-threaded execution model, where each collaborator processes a single (or, a few) thread (track) within a block. 

A fundamental problem in blockchain sharding is how to avoid a (super-)majority of validators in a shard from being adversarial. 
Common solutions to this problem include frequently and randomly re-assigning nodes to different shards; and, hiding the identity of the nodes until after they have proposed their block~\cite{chen2016algorand}.
In our proposed solution, as long as there is at least one collaboration of nodes that are all honest and which collectively operate on all tracks, any invalid transaction or execution in any of the tracks can be caught and the publisher of that block can be slashed.  
Another fundamental problem in sharding is ensure atomicity of cross-shard transactions. 
In our solution, cross-track transactions are executed on a single block providing atomicity. 

The mechanisms needed to prevent ossification as outlined in \S\ref{sec:blockunossify} are necessary for the scaling method in this section to work. 
If blocks are allowed to be collaboratively built without the mechanisms of \S\ref{sec:blockunossify}, a node would be incentivized to join an ossified pool. 
This increases the amount of ossification in the network making the system centralized. 

\section{Interactive Decentralized Applications}

A blockchain is inherently a collaborative application as proposers collaboratively build a ledger of transactions. 
However, the amount of interaction that happens among the nodes during this collaboration is poor. 
A block is constructed by a node independently without any interaction with the other nodes. 
After the block is constructed it is broadcast to the network. 
As a stylized example, if there were a network switch available to which all nodes are connected, then a node would only ever need to interact and address messages to this network switch for broadcast. 
In contrast, collaborations in an ossified coalition often involve a significant amount of direct peer-to-peer and group interactions. 
Our model of a decentralized 
blockchain allows for interactive applications in a natural way. 

Consider a decentralized 
blockchain comprising of free agents with each free agent having certain resources (e.g., compute). 
When free agents collaborate, they pool their resources and interact to provide a service.
Any service that can be provided collaboratively by free agents, and which can be verified by collaborations of other free agents can be realized on the blockchain.  
In general, the more the resources that are pooled in a collaboration, the `better' is the scope and quality of services that can be provided. 

Pooling resources to enhance the service quality is hardly new---many centralized organizations (e.g., cloud providers) do this at scale and with an efficiency that cannot be matched by collaborations whose resources are separated geographically and connected through low-bandwidth public Internet links.  
We claim that trying to mimic the applications offered by centralized providers on a blockchain can see widespread adoption only when the quality of service of the application itself is less important than the application being decentralized.  
For most people and for most applications existing today, this property seems to be not true. 
E.g., the number of users of centralized social media is significantly higher than the number of users using decentralized social media. 

When free agents collaborate they not only pool their physical resources they also pool their human capabilities including the intelligence, skills and expertise that they have. 
These human capabilities can be used, for instance, to inform machine inputs for computing the overall service outputs. 
The presence of a wide range of human capabilities in the system, which can be pooled together through diverse, decentralized collaborations can be a valuable service in many domains to end users. 
Such a service is unlikely to exist through a centralized organization today; or, even if it exists the number, independence and diversity of free agents that can be supported on a blockchain can be significantly greater than those of the free agents operating out of a single centralized organization. 
This results in collaborations that have a lower overall bias and increased innovation in the decentralized service. 

All collaborations have value whether they happen due to ossification or not. 
In an ossified coalition, the outputs of the collaboration are colored by the objective of maximizing efficiency to succeed in the market.  
In an ossified market, coalitions are acutely aware of market trends and  what other coalitions are doing.  
The exit of one coalition from the market is generally favorable to the other coalitions.  
In a decentralized 
collaboration, the outputs of the collaboration are colored by the objective of providing a level of service that is acceptable to a majority of service providers in the overall decentralized system.
For the same type of service, what is acceptable to a majority of providers in a decentralized system can be very different from what is acceptable to a few (e.g., the leadership) in an ossified coalition. 
For example, a decentralized system may place a greater emphasis on shared principles, values, and practices even if those come at the cost of increasing service price or reducing profits.  
In our proposed method, the most important behavior needed to increase rewards is maintaining good trust relationships with other free agents. 
An entity (or a collection of entities) that seeks to compete with others by restricting collaborations to themselves will fail to get a good reward. 
Thus, our system discourages selfish behavior and forces participants to consider the overall wellbeing of the network. 
In certain cases, it may be counterproductive to use a large collaboration set.  
It is possible to limit the maximum size of collaborations allowed on our blockchain.

Collaborations between free agents from domains that don't normally interact can produce novel outcomes of benefit to end users.
While many experts have voiced the increasing need for such collaborations (e.g., see~\cite{alon2026remarks}), in practice lack of adequate economic incentives prevents them from happening.
An ``exploratory'' collaboration is financially risky and is  afforded only by large corporations today. 
A decentralized blockchain can provide a platform where the risk of exploration is absorbed by the other ``exploitative'' collaborations happening on chain, thus, encouraging innovation.  





\section{Preventing On-Chain Merchant Ossification}

Smart contracts have helped create markets for various services in multiple domains including DeFi, social networks, gaming, NFTs etc~\cite{ethereumapps}. 
E.g., prominent services offered by major providers in the DeFi domain comprises of decentralized exchanges, lending and borrowing, liquid staking and restaking services (among others).
If we focus on any one type of service, we see that the market forces once again encourage ossification with a small number of big providers dominating the market~\cite{gogol2024sok}. 
On-chain markets are unique as services are guaranteed to be available as long as the blockchain is live, most services make their code open source, many services are governed by DAOs, and the market itself is permissionless with a low entry barrier. 
Nevertheless, when a large fraction of the market share is held by a single organization, it becomes susceptible to attacks such as token theft~\cite{beanstalk,caldarelli2021blockchain,bzxexploit,oraclemanipattack}, undesirable protocol changes~\cite{CompoundHack}, smart contract bugs~\cite{jiao2024survey} and rugpulls~\cite{anubisdao}.  
The smart contract services landscape can also be decentralized 
to prevent these attacks, and encourage greater collaboration between the service provider community. 

We follow the ideas of \S\ref{sec:blockunossify} to provide a solution sketch. 
Consider a set of providers $P$ providing an on-chain service (e.g., lending).
In a decentralized (Definition~\ref{defn: decentralization}) service, when a user submits a service request it is randomly routed to one of the provider's (e.g., with a probablity proportional to the provider's current market share) service smart contract. 
The chosen provider's smart contract collaborates with other providers' smart contracts to provide the service. 
As in \S\ref{sec:blockunossify}, we can assign an importance value and a market share value (analogous to stake in \S\ref{sec:blockunossify}) to each provider.  
The reward schemes can also be designed analogous to \S\ref{sec:blockunossify}. 
Note that unlike anti-trust laws, our mechanism does not seek to divide a large organization into smaller ones. 
If any of the providers in $P$ are already heavily ossified, we don't try to (nor can we) split it up. 
Instead our protocol merely encourages strong collaboration between the existing providers in $P$. 
In other words, we try to prevent any further ossification from happening in the future, but we cannot do anything about the ossification that has already happened in the past. 

To realize the above solution in today's blockchains, we would require additional smart contracts and methods. 
E.g., we would need a smart contract that receives user requests and directs them to a random provider's contract. 
We would need to keep track of the provider set $P$ and do all the computations associated with calculating the importance.  
Doing all computations on chain may be expensive in existing blockchains.
Since all on-chain activity is public, it may be possible to do the importance computations offline and feed it to the chain via oracles. 
It is also essential that the service operations are clearly defined with standardized interfaces so that multiple providers can  inter-operate seamlessly. 
To derive the full intended benefit of our decentralized service, it is important that providers develop their smart contracts independently without directly copying the source code from a single reference. 

\section{Decentralized Organizations}

Many organizations exist today that are collectively managed by its members---such as worker cooperatives or credit unions---and in the case of blockchains,  decentralized autonomous organizations (DAOs). 
In these collectives, key decisions about how to steer the organization are democratically made while the revenue earned is shared between the members. 
Depending on the organization's structure, a portion of the revenue may be directly awarded to the members; or, the members could indirectly benefit through increasing value of their membership stake (e.g., governance tokens).  

Each of these organizations can, fundamentally, be seen as a market where members compete for votes.
Consider a voting in a DAO for a decision where the possible choices are $A$ or $B$. 
Here, the proponents of $A$ and the proponents of $B$ are the two sellers.
The undecided voters are the buyers.
In many cases, whether $A$ wins or $B$ wins, to a large extent, depends on how effectively their respective proponents are able to convince the undecided voters on the benefits of those choices (i.e., the campaign).
As with other markets we have discussed, campaigning can heavily benefit from ossification. 
When proponents of a choice ($A$ or $B$) pool together their resources (time, money, effort etc.), they can execute a campaign that is much more efficient compared to what an individual member can do alone. 
Therefore, we claim that most---if not all---of the collectives today are internally ossified. 
In particular, measurement studies have shown that DAOs, contrary to their name, are not decentralized and often suffer from strong centralization with  a significant voting share controlled by just a small number of parties.

The ideas we have outlined in this paper allow us to construct collectives that are resistant to ossification and achieve decentralization in the sense of Definition~\ref{defn: decentralization}. 
Members can be incentivized to engage in unossified collaborative discussions about the voting issues at hand. 
Undecided voters can learn about the issues from the collaborations, before deciding what to vote for. 
Our incentives ensure that collaboration sets contain proponents from multiple voting choices, and not just a single choice. 
This provides an opportunity for the undecided voters to receive an unbiased opinion about the voting choices directly from the proponents of the choices. 

In some DAOs, revenue is split proportional to the number of governance tokens a member has. 
In addition to the voting market mentioned above, in these DAOs another market arises: a market for buying and selling governance tokens. 
At any time instant, there may be members that are looking to sell their governance token and others that are looking to purchase them. 
In this market, the sellers are those looking to purchase the governance tokens and the buyers are the those selling their tokens. 
The seller with the best bid (i.e., highest) receives the token. 
Once again, the efficiency of the seller (the bid value) can be improved through ossification. 
By pooling funds, an ossified coalition can purchase a significantly greater number of governance token than what an individual user can. 
Thus, in this market we can expect centralizing clusters of members to arise that own most of the governance tokens. 
If the members of a  coalition are ideologically aligned, the same coalition can engage as a single ossified entity in the token buy-sell market and in the voting campaign market. 
Here too, we can use our ideas to combat ossification. 



\section{Discussion}

Blockchains are trustless---it is sufficient for a user to have trust that the majority of users are honest without requiring to individually trust any user. 
This core principle has informed blockchain design and innovation over the years. 
Even in layer-2 methods where users are required to trust an individual provider (e.g., the sequencer in a rollup), mechanisms are provided through which the provider is slashed for misbehavior. 
Our proposed method is a marked departure from this status quo.  
Not only do we ask nodes to individually trust (a small number of) other nodes to form collaborations, we explicitly slash only the publisher even if it a collaborating node that provably deviated from protocol.
From the various examples highlighted in the paper, we have seen that despite being a trustless system, trust relationships in the form of ossified coalitions almost invariably occur because it is the rational action for users seeking to maximize their payoffs. 
These trust relationships are hard to detect, disincentivize, or ban. 
It is, therefore, natural to try to have mechanisms that encourage the trust relationships to be in a way that is beneficial for the overall system. 


Another key property of blockchains is pseudonymity. 
In our proposed method, users need a platform through which they can discover, interact, and form partnerships with other users. 
Even if the platform supports anonymous interaction, the process of collaboration can reveal private information about a user. 
How to enable users to collaborate while keeping their identities anonymous is an important question. 
Sybil users can flood the communication platform with spam. 
Filtering out the spam is necessary to allow for meaningful trust relationships to form between honest nodes. 

Historically, human societies and organizations have been structured to have a small (relative to the size of the organization),  centralized leadership. 
This is true even in the ossified coalitions occurring on blockchains.
When users are asked to collaboratively perform tasks, it is likely that a similar leadership-worker segregation structure emerges in the system. 
That is, some users may prefer to lead while others prefer to be led. 
This is particularly  important in applications where collaborations require significant human-to-human interaction. 
The impact of such a segregation on the decentralization and value provided by the network is an important aspect that needs to be studied. 

In our proposed method, a user can have a fixed set of collaborators as long as the overall collaboration graph is an expander. 
For certain applications, there may be value in frequently changing the collaboration set of a user over time. 
How to design the reward mechanism so that it looks at not only the connectivity of the collaboration graph, but also its dynamism is an interesting open question. 

\subsection{The Value of Decentralization}

Our work prompts a re-visit to the question of what is the value of decentralization in blockchains.
The common response to this question is: (1) decentralization eliminates single points of failure making the network more resilient to attacks and failures; (2) it provides transparency with verifiable audit trails over an immutable, public ledger; (3) it eliminates the need for trusted intermediaries for executing contracts; (4) it is resistant to censorship as no single entity controls access to the data; (5) it provides credible neutrality where the network's core algorithms apply equally to all participants without discrimination. 
Intuitively, a blockchain that is controlled by just 4 large coalitions with a 1000 validators per coalition is less decentralized than a blockchain with 4000 independent validators.
However, the above five principles are equally applicable to both of these cases as long as a (super-)majority of the validators are honest. 
The argument that the network with 4 coalitions is more susceptible to failures or attacks requires the additional assumption that members of a coalition share a common leadership, infrastructure, software etc. 
If members of a coalition don't share anything in common, except that each member has trust relationships only with other members of the same coalition, it is hard to argue that the scenario with the 4 coalitions is more susceptible to failures or attacks.\footnote{We assume a member is willing to interact only with a member it trusts. A global ``network switch''---to which all members are connected---is trusted by all members.} 
Therefore, there must be a benefit to decentralization beyond the five principles listed above. 

We identify {\em incentivized altruism} as the additional property satisfied by decentralized systems. 
In incentivized altruism, members perform actions that are beneficial to the overall network in exchange for a direct or indirect compensation. 
Unlike regular actions which provide an immediate profit for the member (e.g., mining a block), an altruistic action incurs a non-negative cost (positive or zero) for the member in the short term and may be profitable only in the long run. 
E.g., in today's blockchains (e.g., Bitcoin), nodes help forward blocks and other messages which benefits the entire network. 
Even if the forwarding node is not rewarded directly for its actions, if the node holds tokens it can indirectly benefit by contributing to an increase in the token's value by maintaining a healthy network. 
Users also altruistically contribute to the blockchain by writing open-source software for the clients, wallets, Web-2 interfaces etc. 
Members also help other members through messaging platforms like Discord, posting tutorials, and organizing conferences.
In our proposed method, altruism extends even further as nodes help other nodes in the service execution directly on chain.
On the other hand, in centralized coalitions 
members of a coalition interact with and help only members of the same coalition for service execution. 
The impact of (incentivized) altruism on provider behavior, service quality, and earned rewards is an interesting direction for future work. 

Altruism provide two distinct advantages to the network. 
First, nodes can enhance the amount of service and the quality of service provided by collaborating with other nodes.  
This is particularly useful for newly joined nodes who may not have adequate resources at the beginning. 
Second, altruism enhances the economic {\em stability} of the network. 
A node---even if it is a well-established one---can occasionally experience unforeseen technological or economic challenges (e.g., hardware failure due to a natural disaster, or a sudden financial hardship).
During such times, the altruistic benefactors can help dampen the service impact to end users by helping the affected node. In a way, the benefactors act as a low-premium insurance policy providing valuable assistance at times of need. 

The benefits of collaborative decentralization extend to security as well. 
In a network where nodes have a strong social collaboration, the negative impact of Sybil attacks can be significantly ameliorated. 
E.g., in a proof-of-work chain if a private chain attack originates from a user(s) outside of the core collaboration graph, it can be easily ignored even if it contains more work than the honest chain.  
Such overriding ``social consensus'' decisions have been made by members of various blockchains in the past on a need basis depending on the severity of the security issue.

\subsection{Cross-Chain Collaboration}

Ossification can occur at the level of blockchains, with different blockchains competing to offer the same type of services. 
Even if each blockchain is individually decentralized internally, ossification across chains can cause a small number of blockchains to emerge as the dominant providers with an overwhelming market share (e.g., today Bitcoin holds $> 50\%$ of the market capitalization across cryptocurrencies).
The full benefits of decentralization can be realized only if decentralization is present at all levels. 
This implies there must be a degree of collaboration and altruism present across blockchains, while avoiding the formation of  any blockchain-level coalitions. 
Within an individual blockchain, collaboration and altruism can be codified as law forcing this desired behavior from the users. 
Eliciting such behavior outside of a chain appears challenging. 
The difficulty arises due to two key questions: (1) what is the incentive for a blockchain to altruistically help other blockchains? (2) what does it mean (i.e., how can) for a blockchain to help other blockchains? 

An immediate but unsatisfactory answer to (1) is ``blockchains must help other blockchains because it is in the spirit of (collaborative) decentralization.'' 
Newly built chains can certainly benefit from the assistance of well-established chains. 
However, providing and receiving assistance can benefit even well-established chains. 
If each blockchain collaborates with other blockchains it trusts, the core group of blockchains in the blockchain-level collaboration graph can achieve greater resilience and stability to adverse events that would not be possible without collaboration. 
The core group can also be insulated against catastrophic events occurring outside of the group, such as the crash of a certain token. 
The overall public trust on the core group increases, which benefits all the members of the group. 

A blockchain can help another blockchain in many ways. 
It could provide financial assistance, investments, help with securing the other chain, share knowledge and experience, etc. 
We leave a systematic analysis of these potential solution ideas to both questions for future work. 





\section{Related Work}

Our work has connections to prior research and ideas from various disciplines including economics, sociology, psychology and political science.
We do not attempt to provide an exhaustive list of related works from those fields due to the authors' limited expertise in those areas.
We hope follow-up works from experts in these areas can help bridge any gaps to our work. 
Nevertheless, we provide a small sample of papers that are relevant. 
The concept of ossification, as presented in this paper, is a well-documented phenomenon in economics. 
Autor et al.~\cite{autor2020fall} provides an analysis on the rise of ``superstar'' firms and their impact on the labor share. 
The book Frank et al.~\cite{frank2010winner} discusses winner-take-all markets where even a small difference in performance results in enormous differences in the reward. 
Eisenmann et al.~\cite{eisenmann2006strategies} study the  strong network effects in two-sided markets. 
Haskel et al.~\cite{haskel2017capitalism} present the impact of the infinite-scalability of intangible assets and how they enable superstar firms.  

The necessity of altruism for a well-functioning economy is increasingly being discussed by economists~\cite{park2010altruistic}. 
Becker~\cite{becker1974theory} proposed the famous ``rotten kid theorem'' that shows even selfish members of a family will act for the welfare of the family if the family head is altruistic. 
Andreoni~\cite{andreoni1990impure} introduced the notion of a ``warm glow'' as an intrinsic utility derived by altruists challenging the pure selfless motivations of altruists. 
Fehr et al. \cite{fehr2005economics} rejects the long-held hypothesis that self-interest is the sole motivator for all people. 
It shows how individuals are sometimes willing to sacrifice personal gain to ensure equity, punish wrongdoers, or to help others. 
Bramoull\'e et al.~\cite{bramoulle2024altruism} study the 
production patterns that emerge when people are modeled as altruistic and care about the well being of friends and family. 
It shows how altruistic lending acts as a vital economic cushion during crises and allows the economy to rebound faster when the crises pass. 
Simon~\cite{herbert1957models,simon1990mechanism} argues that if individuals acted purely selfishly, the sheer computational resources required to design perfect, cheat-proof contracts and monitor compliance would overwhelm any economic network.

The idea that collaborating with random or loosely connected people sparks creativity has been well-studied in sociology, economics, and organizational psychology. 
Granovetter~\cite{granovetter1973strength} proposed that casual acquaintances (weak ties) are more valuable than close friends (strong ties) for moving new information and novel opportunities across different social networks. 
Burt~\cite{burt2004structural} identifies ``structural holes''---the empty space between disconnected groups of people---and argues that individuals who bridge these holes are more likely synthesize creative and valuable ideas. 
It is also worth mentioning John Stuart Mills’s~\cite{mill1885principles} opinion that ``it is hardly possible to overrate the value ... of placing human beings in contact with persons dissimilar to
themselves, and with modes of thought and action unlike those with which they are familiar .... Such communication has always been, and is peculiarly in the present age, one of the primary sources of progress.''
Fleming~\cite{fleming2007breakthroughs} proposed that the vast majority of breakthroughs are not entirely new to the world; rather, they are novel, often surprising recombinations of existing technologies or concepts. 
It argues that teams with varied technical expertise and backgrounds are much more likely to generate extreme creative outcomes.

Yokoo et al.~\cite{yokoo2005coalitional} analyze solution concepts for coalition games in a Sybil (termed a ``false name'' attack in the paper) setting. 
They show that the classical Shapley value is not Sybil resistant, and propose a novel solution concept that is Sybil and coalition resistant. 
The paper does not provide an efficient algorithm for computing the proposed solution.
The paper also shows that Sybil resistance can be achieved by defining characteristic functions on ``skills'' (i.e., stake) rather than on agents. 
Later works have developed compact representations and linear-programming based algorithms for computing the solution, but it remains exponential time in the worst case~\cite{ohta2006compact,ohta2008anonymity}.
Lee et al.~\cite{lee2026faithful} provide an efficient approximation algorithm for computing their proposed ``faithful'' Shapley value that is resistant to false-name attacks. 
Roig et al.~\cite{mazorra2023towards} show that there is no solution that simultaneously satisfies efficiency, symmetry, null player, additivity, and Sybil-proofness.

In the domain of blockchains, reputation systems have been used as a basic tool for determining the trustworthiness of users~\cite{hasan2022privacy}.
A reputation score for a user is computed based on the feedback (positive or negative) received about the user from other users. 
Our idea of importance differs from a reputation score in that importance is computed algorithmically based on a user's behavior while reputation is computed from user feedback. 
Decentralized reputation systems use a blockchain for recording and updating user reputation scores~\cite{dennis2015rep,zhou2021blockchain}. 
Key research works in this direction have proposed solutions for dealing with Sybil attacks, collusion attacks, re-entry attacks, and bad-mouthing attacks~\cite{bellini2020blockchain}.
Reputation has also been used as an alternative to stake for developing proof of reputation consensus protocols~\cite{de2020blockchain,zhuang2019proof}. 

A number of recent works have proposed measures for mitigating centralization in blockchains~\cite{kiffer2022centralization}. 
Kiayias et al. study reward sharing schemes modeled as a oceanic game~\cite{kiayias2025pool}. 
Comparing the Shapley mechanism against the
standard proportional scheme, they show that the Shapley value is a competitive alternative that has increased Sybil resistance while providing decentralization.
Middleware solutions such as Eigenlayer~\cite{li2024mitigating,team2024eigenlayer} promote decentralization by allowing validators to use their same staked tokens for securing multiple services. 
Liquid staking protocols (e.g., Lido) allows validators to retain their ETH even while staking them, thereby lowering the barrier to becoming a validator~\cite{scharnowski2025economics}. 
For proof-of-work chains, Miller et al.~\cite{miller2015nonoutsourceable} propose ``non-outsourceable puzzles'' which create a
disincentive for pool operators to outsource mining work. 
Distributed mining pools~\cite{sakurai2025fiberpool} use a decentralized coordinator (e.g., implemented as a smart contract) to replace the centralized coordinator of conventional pools.  

The idea that Sybil nodes have weak social relationships with honest nodes has been used to design Sybil protection systems~\cite{yu2006sybilguard,yu2008sybillimit}. 

\bibliographystyle{splncs04}
\bibliography{main}

\appendix

\section{Importance Transfer in a DAG} \label{apx:imptransfDAG}

\subsection{Preliminaries}
Let us assume all nodes have the same stake for the time being. 
Consider a DAG $D[t]$ of collaborators with a root node $V[t]$. 
The importance of a node $v \in D[t]$ after tax is $I_v[t](1-\beta)$. 
For any subset $S \subseteq D[t]$, define the utility $u(S)$ of $S$ as the total importance of all the nodes reachable from $S$ (including the total importance of $S$ itself), i.e., 
\begin{align} 
u(S) = \sum_{\substack{v \in D[t]: \exists u \in S \text{ such that } \\ v \text{ is reachable from } u}} I_v[t](1-\beta).  \label{eq:utilityfndag}
\end{align}
We define $u(\{ \}) = 0$. 
We have the following proposition. 
\begin{prop}
The utility function $u: 2^{D[t]}\rightarrow \mathbb{R}$ defined in Equation~\eqref{eq:utilityfndag} is a monotone, submodular function. 
\end{prop}
\begin{proof}
Consider any two subsets $S_1, S_2$ with $S_1 \subseteq S_2 \subset D[t]$. 
Let $v \in D[t] \backslash S_2$ be any node not in $S_2$. 
Let $T_1$ be the set of nodes reachable by $v$ (including $v$) but not by $S_1$. 
Similarly, let $T_2$ be the set of nodes reachable by $v$ (including $v$) but not by $S_2$. 
It is easy to see that 
\begin{align}
u(S_i \cup \{v \}) - u(S_i) = \sum_{v' \in T_i} I_{v'}[t](1-\beta),    
\end{align}
for $i=1,2$. 
Consider any $v' \in T_2$. 
This means there is no node from $S_2$ that can reach $v'$. 
Since $S_1 \subseteq S_2$, it follows that there is no node from $S_1$ that can reach $v'$. 
Therefore, $v' \in T_1$. 
Since any node $v' \in T_2$ is also in $T_1$, we have $T_2 \subseteq T_1$. 
Thus, 
\begin{align}
u(S_1 \cup \{ v\}) - u(S_1) \geq u(S_2 \cup \{v \}) - u(S_2),    
\end{align}
which shows that $u$ is submodular. 
We also have that $u(S_1) \leq u(S_2)$ since any node that reachable by $S_1$ is also reachable by $S_2$. 
This shows that $u$ is monotonic. 
\qed 
\end{proof}
We are interested in transferring the importance between different nodes in $D[t]$. 
For any $v \in D[t]$, let $I'_v[t]$ be the importance of $v$ after the transfer. 
To compute how importance is transferred, we consider the convex region defined by the following, 
\begin{align}
\sum_{v \in S} I'_v[t] &\leq u(S),      ~~~~ \forall S \subset D[t] \label{eq:polytope1} \\
\sum_{v \in D[t]} I'_v[t] &= u(D[t]). \label{eq:polytope2}
\end{align}
The inequalities~\eqref{eq:polytope1} and~\eqref{eq:polytope2} define the base of the polymatroid defined by $u$. 
We choose a point on the base that satisfies Properties~\ref{prop:convimp}--\ref{prop:sybil}. 

For any permutation $\sigma$ of the nodes in $D[t]$, the greedy solution given by 
\begin{align}  I'_{v_{\sigma(i)}} = u(\{ v_{\sigma(1)}, v_{\sigma(2)}, \ldots, v_{\sigma(i)} \}) - u(\{ v_{\sigma(1)}, v_{\sigma(2)}, \ldots, v_{\sigma(i-1)} \}), 
\end{align}
for all $i \in \{1,2,\ldots,|D[t]|\}$ is a corner point on the base of the polymatroid. 
Conversely, for any corner point of the base, there exists a permutation whose solution coincides with the corner point~\cite{schrijver2003combinatorial}.
For a permutation $\sigma$, let $I'(\sigma)$ be the solution vector obtained by the greedy algorithm following the order $\sigma$. 
By the convexity of the base of the polymatroid, for any two permutations $\sigma_1$ and $\sigma_2$, the solution $\delta I'(\sigma_1) + (1-\delta) I'(\sigma_2)$ for $\delta \in (0,1)$ is also on the base. 
The Shapley value~\cite{winter2002shapley} is the average of $I'(\sigma)$ averaged over all possible permutations.  
In addition to being computationally expensive, the Shapley value does not guarantee that the solution computed is fair or Sybil resistant. 

Our solution uses the DAG topology to derive a weighting for the permutations such that the resultant solution is fair and Sybil resistant. 
For any node $v \in D[t]$, let $\Gamma(v)$ denote the children of $v$ in $D[t]$. 
We hierarchically construct the permutations and their weights as follows. 
\begin{enumerate}
\item We first partition $D[t]$ into two sets: the root $V[t]$ and the rest of the DAG $D[t]\backslash V[t]$. 
The sequence $(V[t], D[t] \backslash V[t])$ receives a weight of $\alpha$, while the sequence $(D[t]\backslash V[t], V[t])$ receives a weight of $1-\alpha$. 
The set $D[t]\backslash V[t]$ will be further partitioned and ordered in the subsequent steps with additions weights. 
\item Let $v_1, v_2, \ldots, v_{|\Gamma(V[t])|}$ be the nodes in $\Gamma(V[t])$. 
For any $v_i$ for $i=1,2,\ldots,|\Gamma(V[t])|$, let $R(v_i)$ be the set of nodes reachable from $v_i$ in $D[t]$ including $v_i$. 
Consider any permutation $\sigma$ of the nodes in $\Gamma(V[t])$. 
We partition and order $D[t]\backslash V[t]$ as $(R(v_{\sigma(1)}), R(v_{\sigma(2)}) \backslash R(v_{\sigma(1)}), R(v_{\sigma(3)}) \backslash \{ R(v_{\sigma(1)}) \cup R(v_{\sigma(2)})\}, 
\ldots)$ with each permutation receiving a weight of $1/|\Gamma(V[t])|!$. 
If any of $R(v_{\sigma(i)}) \backslash \{ R(v_{\sigma(1)}) \cup R(v_{\sigma(2)}) \ldots R(v_{\sigma(i-1)})\}$ is empty we omit that entry in the ordering. 
\item For any $i=1,2,\ldots, |\Gamma(V[t])|$, the subgraph induced by the vertices $R(v_{\sigma(i)}) \backslash \{ R(v_{\sigma(1)}) \cup R(v_{\sigma(2)}) \ldots R(v_{\sigma(i-1)})\}$ is a DAG with root $v_{\sigma(i)}$. 
We repeats steps 1 and 2 for all the $|\Gamma(V[t])|$ DAGs. 
\end{enumerate}

\subsection{Simplifying the Algorithm} \label{sec:simplifyalg}

Consider a node $v \in D[t]$ and any permutation $\sigma$ of the nodes in $D[t]$. 
In $\sigma$, let $p(v)$ be the first appearing predecessor node of $v$ (i.e., any node that can reach $v$ in $D[t]$). 
If there are no predecessor nodes appearing before $v$ in $\sigma$, we set $p(v) = \phi$ (null).
If $p(v)$ is not null for $\sigma$, then in the greedy solution for $\sigma$ the importance of $v$ is counted towards the solution for $p(v)$, i.e., $I'_{p(v)}[t]$ for $\sigma$ is a summation of importance of nodes one of which is $v$. 
If $p(v)$ is null, the importance of $v$ is counted towards the solution for $v$. 
In any of our permutations, by construction, the importance of $v$ can only be counted towards the solution for $v$ or any of $v$'s predecessors. 

For any specific predecessor $u$ of $v$, let $Q(u) \subset D[t]$ be the nodes in $D[t]$ than can reach $u$ (excluding $u$). 
For $p(v)$ to be equal to $u$ in $\sigma$, none of the nodes in $Q(u)$ must appear before $u$ in the permutation. 
Consider any path $V[t], v_1, v_2, \ldots, v_k, u$ from the root to $u$. 
The chance of the root $V[t]$ appearing after $u$ in the permutation is $(1-\alpha)$. 
The chance of both $V[t]$ and $v_1$ appearing after $u$ in the permutation is $(1-\alpha)^2/|Q(u) \cap \Gamma(V[t])|$. 
Similarly, for any $v_i$ on the path, the chance of $V[t], v_1, \ldots, v_i$ appearing after $u$ in the permutation is $(1-\alpha)^{i+1}/(|Q(u) \cap \Gamma(V[t])| * |Q(u) \cap \Gamma(v_1)| * \ldots * |Q(u) \cap \Gamma(v_{i-1})|)$. 
Finally the chance of $u$ appearing before $v$ but $V[t], v_1,\ldots, v_k$ appearing after $v$ in the permutation is $(1-\alpha)^{k+1}\alpha/(|Q(u) \cap \Gamma(V[t])| * |Q(u) \cap \Gamma(v_1)| * \ldots * |Q(u) \cap \Gamma(v_{k})|)$. 

Let $\mathcal{P}(u)$ be the set of all paths from the root $V[t]$ to $u$ in $D[t]$. 
The total fraction of importance that node $v$ ``pays'' to node $u$ in the asymmetric Shapley value solution is 
\begin{align}
\sum_{(V[t], v_1, \ldots, v_k, u) \in \mathcal{P}(u)} \frac{(1-\alpha)^{k+1}\alpha}{(|Q(u) \cap \Gamma(V[t])| * |Q(u) \cap \Gamma(v_1)| * \ldots * |Q(u) \cap \Gamma(v_{k})|)}. \label{eq:importancetransferuv}
\end{align}
Instead of enumerating all permutations and explicitly computing the greedy solutions and then averaging them to compute the Shapley value solution, we can instead adopt the following simpler but equivalent solution motivated by Equation~\eqref{eq:importancetransferuv}. 
Recall $I_v[t](1-\beta)$ is the initial importance of a node $v$ after tax but before the DAG importance transfer. 
Consider the DAG $D_{Q(v)}[t]$ induced by $Q(v)$ on $D[t]$. 
For any edge $(u, u') \in D_{Q(v)}[t]$ associate an edge weight $w(u,u') = (1-\alpha)/|\Gamma(u)|$. 
The payment algorithm for any node $v \in D[t]$ is as follows. 
\begin{enumerate}
\item All edges in $D_{Q(v)}[t]$ are marked ``unvisited''. 
The importance $i_u$ received by any node $u$ in $D_{Q(v)}[t]$ from $v$ is set to 0. 
\item $v$ pays an amount $I_v[t](1-\beta) \alpha$ to the root $V[t]$, i.e., $i_{V[t]} \leftarrow I_v[t](1-\beta) \alpha$.  
\item Consider any unvisited edge $(u, u') \in D_{Q(v)}[t]$ such that all incoming edges to node $u$ are visited.
Let $i_u$ be the total importance received by $u$ from $v$. 
The importance received by $u'$ from $v$ is incremented as $i_{u'} \leftarrow i_{u'} + i_u * w(u, u')$. 
Mark edge $(u, u')$ as visited.
\item Repeat step 3 until all edges in $D_{Q(v)}[t]$ are visited. 
\end{enumerate}

\subsection{Accounting for Stake} \label{sec:accforstake}

The algorithm presented in \S\ref{sec:simplifyalg} assumes all nodes have the same stake. 
In this section, we explain how the previous algorithm can be modified if nodes have heterogeneous stake. 
For any node $v \in D[t]$, let $s_v$ be the stake of $v$. 
Without loss of generality, assume $s_v$ is an integer multiple of a small constant $\epsilon$ for all $v$. 
Given a collaboration dag $D[t]$ where the nodes have heterogeneous stake, we convert it into another DAG $D'[t]$ where the nodes have homogeneous stake. 
Consider a node $v \in D[t]$ with stake $s_v$. 
Let $\Gamma_{\text{in}}(v)$ be the set of parents of $v$ and $\Gamma(v)$ is the set of children of $v$ in $D[t]$. 
To compute $D'[t]$ we replace each node node $v$ by a chain of $s_v / \epsilon$ nodes with a total importance of $I_v (1 - \beta)$ (how the importance is distributed among the nodes in the chain does not matter).
The last node of $v$'s chain has outgoing edges to the first node of the chains of $v$'s children in $D[t]$. 
We then apply the algorithm of \S\ref{sec:simplifyalg} on $D'[t]$. 
The modified algorithm for $D[t]$ is given below. 
\begin{enumerate}
\item All edges in $D_{Q(v)}[t]$ are marked ``unvisited''. 
The importance $i_u$ received by any node $u$ in $D_{Q(v)}[t]$ from $v$ is set to 0. 
\item $v$ pays an amount $I_v[t](1-\beta) (1 - (1-\alpha)^{s_{V[t]}/\epsilon} ) $ to the root $V[t]$, i.e., $i_{V[t]} \leftarrow I_v[t](1-\beta) (1 - (1-\alpha)^{s_{V[t]}/\epsilon} )$.  
\item Consider any unvisited edge $(u, u') \in D_{Q(v)}[t]$ such that all incoming edges to node $u$ are visited.
Let $i_u$ be the total importance received by $u$ from $v$. 
The importance received by $u'$ from $v$ is incremented as 
\begin{align}
i_{u'} \leftarrow i_{u'} + \frac{i_u (1-\alpha)^{s_u/\epsilon}}{1-(1-\alpha)^{s_u/\epsilon}}  * \frac{(1-(1-\alpha)^{s_{u'}/\epsilon})}{|\Gamma(u)|}.
\end{align}
Mark edge $(u, u')$ as visited.
\item Repeat step 3 until all edges in $D_{Q(v)}[t]$ are visited. 
\end{enumerate}

\subsection{Proof of Theorem~\ref{thm:importancetransfer}} 
\label{apx:proofofthm1}
\begin{proof}
Since our importance transfer solution is a point on the base of the submodular polymatroid, Property~\ref{prop:convimp} is trivially true. 

Next, we show Property~\ref{prop:fairness}. 
Consider any path-closed subset $S \subset D[t]$ containing $V[t]$.   
This means for any node $v \in S$ all the predecessors of $v$ are also in $S$. 
By the algorithm in \S\ref{sec:accforstake}, $v$ transfers its importance to its predecessors, which are nodes in $S$. 
None of the nodes $S$ transfer their importance to a node outside of $S$. 
Therefore, $\sum_{v \in S}(I'_v[t] - I_v[t](1-\beta )) \geq 0$. 
On the other hand, the nodes in $D[t] \backslash S$ transfer importance to their predecessors in $S$. 
However, the nodes in $D[t] \backslash S$ do not receive any additional importance from any node outside of $D[t] \backslash S$. 
Hence, $\sum_{v \in D[t]\backslash S} (I'_v[t] - I_v[t](1-\beta)) \leq 0$.

To prove Equation~\eqref{eq:fariness2}, consider the predecessors $Q(v)$ of any node $v \in D[t]$. 
Let $u \in Q(v)$ be any node. 
The fraction of importance transferred from $v$ to $u$ is 
\begin{align}
&\sum_{\text{path }(V[t], u_1, \ldots, u_k, u)) \in D[t]} \frac{(1-\alpha)^{\frac{s_{V[t]} + \sum_{i=1}^k s_{u_i}}{\epsilon}}(1-(1-\alpha)^{\frac{s_u}{\epsilon}})}{|Q(u) \cap \Gamma(V[t])|*|Q(u) \cap \Gamma(u_1)| * \ldots * |Q(u) \cap  \Gamma(u_k)|} \notag \\
&\leq \sum_{\text{path }(V[t], u_1, \ldots, u_k, u)) \in D[t]} \frac{(1-(1-\alpha)^{\frac{s_u}{\epsilon}})}{|Q(u) \cap \Gamma(V[t])|*|Q(u) \cap \Gamma(u_1)| * \ldots * |Q(u) \cap  \Gamma(u_k)|} \notag \\
&= (1-(1-\alpha)^{\frac{s_u}{\epsilon}}). \label{eq:impupperbound}
\end{align}
For any node $u \in D[t] \backslash Q(v)$, the importance transferred from $v$ is 0. 
Hence, 
\begin{align}
I'_v[t] &\geq I_v[t](1-\beta)(1 - \sum_{u \in Q(v)} (1 - (1-\alpha)^\frac{s_u}{\epsilon}) )   \notag \\
&\geq I_v[t](1-\beta )(1 - \sum_{u \in Q(v)} (1 - (1-\frac{\alpha s_u}{\epsilon})) ) \notag \\
&= I_v[t](1-\beta )(1 - \frac{\alpha}{\epsilon}\sum_{u \in Q(v)} s_u  ).  
\end{align}

Next, we show Property~\ref{prop:slowtransf}. 
For any two nodes $u, v \in D[t]$, by the same argument as used for Equation~\eqref{eq:impupperbound}, the total fraction of importance transferred from $v$ to $u$ is at most $(1-(1-\alpha)^\frac{s_u}{\epsilon}) \leq \alpha s_u / \epsilon$. 
The total amount of importance received by $u$ from $v$ is at most $I_v[t](1-\beta) \alpha s_u /\epsilon$. 
The total amount of importance received by $u$ is bounded as
\begin{align}
\sum_{v \in D[t]} I_v[t](1-\beta) \frac{\alpha}{\epsilon} s_u &\leq \frac{\alpha}{\epsilon} \sum_{v \in D[t]} I_v[t](1-\beta) s_u \notag \\
& \leq \frac{\alpha}{\epsilon} \left( \sum_{v \in D[t]} I_v[t](1-\beta) \right) \left( \sum_{v' \in D[t]} s_{v'} \right). 
\end{align}

Next, we show Property~\ref{prop:sybil}. 
The fraction of importance received by $v$ before being replaced by a DAG from any node $u \in R(v)$ is 
\begin{align}
\sum_{\text{path }(V[t], u_1, \ldots, u_k, v)) \in D[t]} \frac{(1-\alpha)^{\frac{s_{V[t]} + \sum_{i=1}^k s_{u_i}}{\epsilon}}(1-(1-\alpha)^{\frac{s_v}{\epsilon}})}{|Q(v) \cap \Gamma(V[t])|*|Q(v) \cap \Gamma(u_1)| * \ldots * |Q(v) \cap  \Gamma(u_k)|}. \label{eq:importancefractov}
\end{align}
Let $A[t]$ be the DAG that node $v$ is replaced as and $w$ a parent of node $v$ in $D[t]$. 
If $v$ is replaced by $A[t]$ and $u_0$ is a root of $A[t]$, for any path $(V[t], u_1, \ldots, u_k, v)$ to $v$ in the original DAG, the term $(1-(1-\alpha)^\frac{s_v}{\epsilon})/|Q(v)\cap \Gamma(u_k)|$ in Equation~\eqref{eq:importancefractov} is replaced by 
\begin{align}
\sum_{v' \in A[t] } \sum_{(u_0, u_1, \ldots, u_k, v')) \in A[t]} &\frac{(1-\alpha)^{\frac{ \sum_{i=0}^k s_{u_i}}{\epsilon}}(1-(1-\alpha)^{\frac{s_{v'}}{\epsilon}})}{|Q(v') \cap \Gamma(w) ||Q(v') \cap \Gamma(u_0)|* \ldots * |Q(v') \cap \Gamma(u_k)|}  \notag \\
& \leq \sum_{v' \in A[t]} \frac{(1-(1-\alpha)^\frac{s_{v'}}{\epsilon})}{|Q(v') \cap \Gamma(w)|} \leq \sum_{v' \in A[t]} \frac{(1-(1-\alpha)^\frac{s_{v'}}{\epsilon})}{|Q(v) \cap \Gamma(w)|} \notag \\
& \leq \frac{(1-(1-\alpha)^\frac{s_v}{\epsilon}) + \frac{\alpha^2}{2}(\frac{s_v^2}{\epsilon^2} - \sum_{v' \in A[t]}\frac{s_{v'}^2}{\epsilon^2}) + O(\alpha^3)}{|Q(v) \cap \Gamma(w)|} \label{eq:taylorseries} \\ 
& \leq \frac{(1-(1-\alpha)^\frac{s_v}{\epsilon}) + \frac{\alpha^2}{2}(\frac{s_v^2}{\epsilon^2} -\frac{s_{v}}{\epsilon}) + O(\alpha^3)}{|Q(v) \cap \Gamma(w)|} ,  
\end{align}
where inequality~\eqref{eq:taylorseries} follows from the Taylor series approximation for the function $f(x) = (1 - x)^\frac{s_{v'}}{\epsilon}$. 
Therefore, the maximum factor by which the fraction of importance received by $v$ can increase is at most 
\begin{align}
    1 + \frac{\frac{\alpha^2}{2}(\frac{s_v^2}{\epsilon^2} - \frac{s_v}{\epsilon}) + O(\alpha^3)}{(1-(1-\alpha)^\frac{s_v}{\epsilon}))},
\end{align}
which tends to 0 as $\alpha \rightarrow 0$. 
\qed 
\end{proof}

\subsection{Non-asymptotic Sybil Resistance at the Cost of Fairness.}

Consider the following approximation to the importance transfer algorithm presented in \S\ref{sec:accforstake}.
\begin{enumerate}
\item All edges in $D_{Q(v)}[t]$ are marked ``unvisited''. 
The importance $i_u$ received by any node $u$ in $D_{Q(v)}[t]$ from $v$ is set to 0. 
\item $v$ pays an amount $I_v[t](1-\beta) \frac{\alpha s_{V[t]}}{\epsilon}$ to the root $V[t]$, i.e., $i_{V[t]} \leftarrow I_v[t](1-\beta) \frac{\alpha s_{V[t]}}{\epsilon}$.  
\item Consider any unvisited edge $(u, u') \in D_{Q(v)}[t]$ such that all incoming edges to node $u$ are visited.
Let $i_u$ be the total importance received by $u$ from $v$. 
The importance received by $u'$ from $v$ is incremented as 
\begin{align}
i_{u'} \leftarrow i_{u'} + \frac{i_u \epsilon}{\alpha s_u}  * \frac{\alpha s_{u'}}{\epsilon |\Gamma(u)|}.
\end{align}
Mark edge $(u, u')$ as visited.
\item Repeat step 3 until all edges in $D_{Q(v)}[t]$ are visited. 
\end{enumerate}
For the same DAG $D[t]$, the fraction of importance transferred by a node $v \in D[t]$ to its predecessors is slightly greater in this algorithm compared to the algorithm in \S\ref{sec:accforstake}.   
In the algorithm of \S\ref{sec:accforstake}, if $I_v[t] > 0$ for a node $v$, then its importance $I'_v[t]$ after the transfer is also guaranteed to be positive. 
However, in the algorithm above the importance of a node can become zero for a sufficiently large value of $\alpha$. 
Hence, the algorithm above is more unfair compared to the one in \S\ref{sec:accforstake}. 
However, unlike the algorithm in \S\ref{sec:accforstake} which only has asymptotic sybil resistance (Property~\ref{prop:sybil}), the algorithm above has sybil resistance for any constant $\alpha > 0$. 
The algorithm above is also more computationally efficient compared to the algorithm in \S\ref{sec:accforstake}. 
In practice, the algorithm presented above may be used with a reasonable value for $\alpha$ to guarantee that the importance of nodes at the bottom of the DAG do not go to zero. 
E.g., if it is known (or, enforced) that the total stake of any DAG $D[t]$ is at most $s$ for $0 < s < 1$, we can use $\alpha = \epsilon/(2s)$. 

\subsection{Example}

\begin{figure}[!tbp]
    \begin{subfigure}[t]{0.3\textwidth}
        \centering
        \includegraphics[width=.6\textwidth]{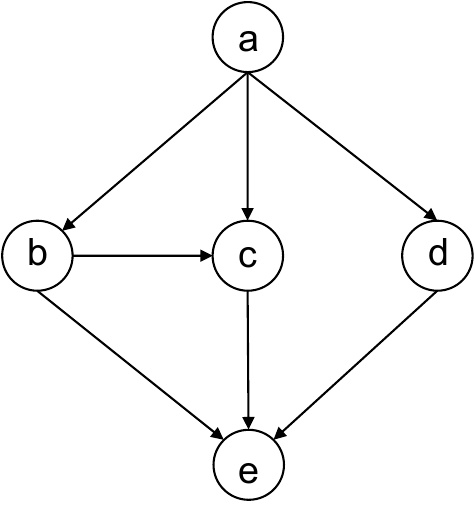}
        \caption{}
        \label{fig:dag1}
    \end{subfigure}
    \hfill
    \begin{subfigure}[t]{0.35\textwidth}
        \centering
        \includegraphics[width=.9\textwidth]{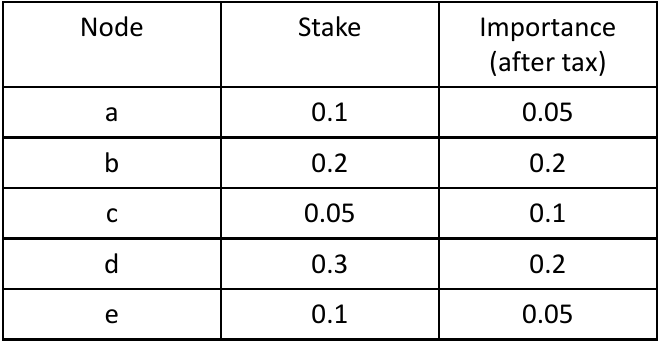}
        \caption{}
        \label{fig:dag2}
    \end{subfigure}
    \hfill
    \begin{subfigure}[t]{0.3\textwidth}
        \centering
        \includegraphics[width=.7\textwidth]{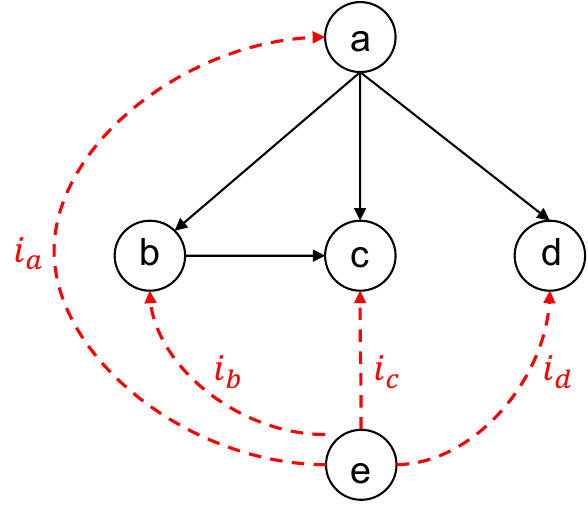}
        \caption{}
        \label{fig:dag3}
    \end{subfigure}    
    \caption{(a) Collaboration DAG $D[t]$. (b) Stake and importance of nodes initially after tax. (c) Node $e$ transfers importance to all its ancestors in $D[t]$. }
    \label{fig:dag}
\end{figure}

Consider the DAG shown in Fig.~\ref{fig:dag1}. 
The stake and the initial importance of the nodes after tax is shown in Fig.~\ref{fig:dag2}. 
Consider node $e$. 
The ancestors of node $e$ are the nodes $a, b, c$ and $d$. 
The fractions of importance that node $e$ transfers to each of those nodes are given by: 
\begin{align}
i_a &= 1-(1-\alpha)^{s_a/\epsilon} \notag \\
i_b &= i_a \frac{ (1-\alpha)^{s_a/\epsilon} (1-(1-\alpha)^{s_b/\epsilon})}{3(1-(1-\alpha)^{s_a/\epsilon})}\notag \\
i_c &= i_a \frac{(1-\alpha)^{s_a/\epsilon} (1-(1-\alpha)^{s_c/\epsilon})}{3(1-(1-\alpha)^{s_a/\epsilon})} + i_b \frac{(1-\alpha)^{s_b/\epsilon}  (1-(1-\alpha)^{s_c/\epsilon})}{2 (1-(1-\alpha)^{s_b/\epsilon})}     \notag \\
i_d &= i_a \frac{ (1-\alpha)^{s_a/\epsilon} (1-(1-\alpha)^{s_d/\epsilon})}{3 (1-(1-\alpha)^{s_a/\epsilon})}.
\end{align}
Using $\epsilon = 0.001$ and $\alpha = \epsilon/2$, we have 
$i_a = 0.049, i_b = 0.030, i_c = 0.011, i_d = 0.044$. 
Thus, node $e$ transfers a total fraction of $ 0.134$ of its importance to its ancestor nodes. 
We can similarly compute the fraction of importance transferred by each of the nodes $b, c$ and $d$ to their ancestors. 
Node $b$ transfers a fraction $0.049$ of its importance to node $a$. 
Node $c$ transfers a fraction $0.049$ to node $a$ and a fraction $0.030$ to node $b$. 
Node $d$ transfers a fraction $0.049$ to node $a$. 
The final importance of nodes $a,b,c,d$ and $e$ after the transfer is $0.077, 0.195, 0.093,0.192$ and $0.043$ respectively. 

\section{Analysis}

\subsection{Proof of Theorem~\ref{thm:preventingossification}}
\label{s:proofofthmprevossification}

\begin{proof}
Since $\sum_{(u,v): u \in S, v \in V\backslash S}w_{(u,v)}[t] < \phi \sum_{u \in S} I'_u[t]$, the set $S$ is excluded from the maximal expander subgraph $G'[t]$ (if it exists) for all $t$.
Therefore, the nodes in $S$ effectively lose $\sum_{u \in S}I_u[t] \beta$ importance in tax each round, without the importance being refunded at the end of the round. 
The maximum amount of importance available with $V\backslash S$ is $2 - \sum_{u \in S}I_u[t](1-\beta)$ after tax. 
Therefore, the maximum amount of importance that $S$ can receive in a round from $V\backslash S$ is at most  $\alpha (\sum_{u\in S}s_u )(2 - \sum_{u \in S}I_u[t](1-\beta)) / \epsilon$.  
Assuming a large initial importance for $S$, the importance of $S$ continues to decrease until the importance lost by $S$ due to tax is smaller than the importance gained by $S$ from $V\backslash S$. 
That is, 
\begin{align}
\sum_{u \in S} I_u[t] \beta \leq  \frac{\alpha}{\epsilon} (\sum_{u\in S}s_u )(2 - \sum_{u \in S}I_u[t](1-\beta)), \\
\implies \sum_{u \in S} I_u[t] \leq \frac{2\alpha \sum_{u\in S}s_u}{\beta(\epsilon - \alpha\sum_{u\in S}s_u) + \alpha \sum_{u\in S}s_u}. 
\end{align}
\qed
\end{proof}

\subsection{Proof of Lemma~\ref{lem:iclowerbound}}
\label{s:proofoflemiclowerbound}

\begin{proof}
At the beginning of each round the collector gains $\sum_{v\in V}I_v[t]\beta$ importance from taxes. 
At the end of a round the collector loses 
\begin{align}
\sum_{v\in G'[t] } \frac{ I'_v[t] (\sum_{v \in G'[t]} s_v) \beta (I_c[t] + (2 - I_c[t])\beta )}{\sum_{v \in G'[t]} I'_v[t]} = (\sum_{v \in G'[t]} s_v)  \beta (I_c[t] + (2 - I_c[t])\beta )
\end{align}
importance to $G'[t]$. 
Therefore, 
\begin{align}
I_c[t+1] = I_c[t] +  \sum_{v\in V}I_v[t]\beta - (\sum_{v \in G'[t]} s_v) \beta (I_c[t] + (2 - I_c[t])\beta ), \label{eq:importancecollectordynamics}
\end{align}
for all $t\geq 0$. 
Since importance is conserved, we have $\sum_{v \in V}I_v[t] = 2 - I_c[t]$ for all $t\geq 0$. 
Therefore, Equation~\eqref{eq:importancecollectordynamics} becomes
\begin{align}
I_c[t+1] = I_c[t] + (2 - I_c[t])\beta - (\sum_{v \in G'[t]} s_v) \beta (I_c[t] + (2 - I_c[t])\beta ).  
\end{align}
We have $I_c[t+1] \geq I_c[t]$ if $I_c[t] \leq (2-2\beta \sum_{v\in G'[t]}s_v)/(1+\sum_{v\in G'[t]}s_v(1-\beta))$. 
Since, 
\begin{align}
\frac{2-2\beta \sum_{v\in G'[t]}s_v}{1+\sum_{v\in G'[t]}s_v(1-\beta)} \geq  \frac{2-2\beta}{2-\beta},  
\end{align}
if $I_c[t]$ drops below $(2-2\beta)/(2-\beta)$, $I_c[t+1]$ is guaranteed to be bigger than $I_c[t]$. 
If $I_c[t]$ is slightly above the threshold of $(2-2\beta)/(2-\beta)$, $I_c[t+1]$ can 
undershoot below the threshold by a maximum of $2\beta$ after which it must increase. 
Therefore, we conclude that 
\begin{align}
I_c[t] \geq     \frac{2-2\beta}{2-\beta} - 2\beta \approx 1,
\end{align}
for all time $t>0$. 
\qed 
\end{proof}

\subsection{Proof of Theorem~\ref{thm:rewardingdecentralization}}
\label{s:proofofthemrewardingdecent}

\begin{proof}
Consider the subgraph $G[t](S\cup S')$ of $G[t]$ spanned by $S \cup S'$ for any $s' \subseteq V \backslash S$. 
The conductance of the cut $S'$ in the subgraph is at most 
\begin{align}
\frac{\phi_\mathrm{out}\min(\sum_{v\in S}I'_v[t], \sum_{v\in S' }I'_v[t])}{\min(\sum_{v\in S}I'_v[t], \sum_{v\in S'}I'_v[t])} < \phi. 
\end{align}
Therefore, $G'[t](S\cup S')$ cannot be a $\phi$-expander for any $S' \subseteq V \backslash S$.
On the other hand, the subgraph spanned by $S$ is an expander with a total stake $> 0.5$. 
Hence, $G'[t] = G[t](S)$. 

At time $t$, $S$ loses $\sum_{v\in S}I_v[t]\beta$ importance to the collector. 
It gains $(\sum_{v\in S}s_v) \beta (I_c[t] + (2-I_c[t])\beta)$ importance back as a tax refund. 
It also loses $\sum_{(u,v): u\in S, v\in V\backslash S}I_{(u,v)}[t]$ importance to $V \backslash S$ through collaborations (if any). 
Hence, 
\begin{align}
\sum_{v\in S}I_v[t+1] = \sum_{v \in S} I_v[t](1-\beta) + (\sum_{v\in S}s_v) \beta (I_c[t] + (2-I_c[t])\beta) - \sum_{(u,v): u\in S, v\in V\backslash S}I_{(u,v)}[t] \notag \\
\geq \sum_{v \in S} I_v[t](1-\beta) + (\sum_{v\in S}s_v) \frac{\beta(2 - 4\beta + 6\beta^2 -2\beta^3)}{2-\beta} - \sum_{(u,v): u\in S, v\in V\backslash S}I_{(u,v)}[t], \label{eq:ilowerbound}
\end{align}
where the second inquality follows from Lemma~\ref{lem:iclowerbound}.  
From inequality~\eqref{eq:ilowerbound} above, we have 
\begin{align}
\sum_{v\in S} w_v[t+1] - \alpha (1-\alpha)^{t+1} I_v[0] \geq \sum_{v\in S} w_v[t] (1-\beta) \notag \\
+ (\sum_{v\in S}s_v) \frac{\beta(2 - 4\beta + 6\beta^2 -2\beta^3)}{2-\beta}(1-(1-\alpha)^{t+1}) - \sum_{(u,v):u\in S, v\in V\backslash S} w_{(u,v)}[t] \\ 
\geq \sum_{v\in S} w_v[t] (1-\beta) \notag \\ 
+ (\sum_{v\in S}s_v) \frac{\beta(2 - 4\beta + 6\beta^2 -2\beta^3)}{2-\beta}(1-(1-\alpha)^{t+1}) - \phi_\mathrm{out}\sum_{v\in S} I'_v[t] \label{eq:middle1}\\
\geq \sum_{v\in S} w_v[t] (1-\beta) \notag \\
+ (\sum_{v\in S}s_v) \frac{\beta(2 - 4\beta + 6\beta^2 -2\beta^3)}{2-\beta}(1-(1-\alpha)^{t+1}) - \phi_\mathrm{out}, \label{eq:middle2}
\end{align}
where inequality~\eqref{eq:middle1} follows due to the weak collaboration between $S$ and $V\backslash S$ and inequality~\eqref{eq:middle2} follows from Lemma~\ref{lem:iclowerbound}. 
Taking $\liminf$ on both sides of the inequality above and using the property that $\liminf_{t\rightarrow \infty}(x[t] + y[t]) \geq \liminf_{t\rightarrow\infty}x[t] + \liminf_{t\rightarrow \infty}y[t]$ for any two sequences $x[t]$ and $[t]$, we have 
\begin{align}
\liminf_{t\rightarrow\infty} \sum_{v\in S}  w_v[t] \geq (1-\beta)\liminf_{t\rightarrow\infty} \sum_{v\in S} w_v[t] + (\sum_{v\in S}s_v) \frac{\beta(2 - 4\beta + 6\beta^2 -2\beta^3)}{2-\beta} - \phi_\mathrm{out} \notag \\
\implies \liminf_{t\rightarrow\infty} \sum_{v\in S} w_v[t] \geq (\sum_{v\in S}s_v) \frac{(2 - 4\beta + 6\beta^2 -2\beta^3)}{2-\beta} - \frac{\phi_\mathrm{out}}{\beta}. 
\end{align}
\qed 
\end{proof}

\end{document}